\documentclass[a4paper,11pt]{article}

\usepackage{microtype}
\usepackage{mathtools,amsmath}
\usepackage{amssymb,amsthm}
\usepackage[T1]{fontenc}
\usepackage{graphicx}
\usepackage{xspace}
\usepackage{hyperref,color}
\usepackage{pdfpages}
\usepackage{fullpage} 
\usepackage[mathlines]{lineno}
\usepackage{enumerate}
\usepackage{mathrsfs}
\usepackage[linesnumbered,ruled,vlined]{algorithm2e}

 \newtheorem{theorem}{Theorem} \newtheorem{lemma}[theorem]{Lemma}
 
 \newtheorem{corollary}[theorem]{Corollary}
 
 \newtheorem{definition}[theorem]{Definition}
 \newtheorem{observation}[theorem]{Observation}

\newcommand{\xcoord}{\ensuremath{x}}
\newcommand{\ycoord}{\ensuremath{y}}
\newcommand{\myint}[1]{\ensuremath{\textsf{int}(#1)}}

\newcommand{\mybd}[1]{\ensuremath{\partial#1}}
\newcommand{\mycl}[1]{\ensuremath{\textsf{cl}(#1)}}

\newcommand{\vv}[1]{\ensuremath{\mathbf{#1}}}

\newcommand{\usetp}{\ensuremath{\mathbb{S}^+}}

\newcommand{\prob}[2]{\ensuremath{\textsc{OSP}_{#1}^{\mathrm{#2}}}}

\newcommand{\nnprob}{\ensuremath{\textsc{OSP}}}

\newcommand{\nopt}{\ensuremath{\textsf{opt}}}
\newcommand{\dwidth}[2]{\ensuremath{\omega_{#1}(#2)}}

\newcommand{\projx}[1]{\ensuremath{I_x(#1)}}

\newcommand{\ac}[1]{\mathsf{AC}(#1)}   
\newcommand{\iant}{\mathsf{Ant}}   
\newcommand{\td}{\ensuremath{V}}   
\newcommand{\normal}{\ensuremath{\vec{\eta}}}
\newcommand{\filter}{\ensuremath{ \mathrm{Filter}_{1}}}

\newcommand{\cinterval}{\ensuremath{\mathbf{I}}}
\newcommand{\ccb}{\ensuremath{\mathscr{E}}}
\def\denseitems{
    \itemsep1pt plus1pt minus1pt
    \parsep0pt plus0pt
    \parskip0pt\topsep0pt}

\newbox\ProofSym \setbox\ProofSym=\hbox{%
  \unitlength=0.18ex%
  \begin{picture}(10,10) \put(0,0){\framebox(9,9){}}
    \put(0,3){\framebox(6,6){}}
  \end{picture}}

\title{Orthogonal Strip Partitioning of Polygons: Lattice-Theoretic Algorithms and Lower Bounds\thanks{This work was supported in part by a KIAS Individual Grant AP106101 via the Center for Artificial Intelligence and Natural Sciences at Korea Institute for Advanced Study, and by the Center for Advanced Computation at Korea Institute for Advanced Study.
}}
\author{Jaehoon Chung\thanks{Korea Institute for Advanced Study (KIAS), Seoul, Korea. 
{\tt sk7755@kias.re.kr}}}
\begin{document}
\date{}
\maketitle

\begin{abstract}
We study a variant of a polygon partition problem, 
introduced by Chung, Iwama, Liao, and Ahn [ISAAC'25]. 
Given orthogonal unit vectors $\mathbf{u},\mathbf{v}\in \mathbb{R}^2$ and a polygon $P$ with $n$ vertices, 
we partition $P$ into 
connected pieces by cuts parallel to $\mathbf{v}$ 
such that each resulting subpolygon has width at most one in direction $\mathbf{u}$.
We consider the value version, 
which asks for the minimum number of strips, 
and the reporting version, which outputs a compact encoding of the cuts in an optimal strip partition.

We give efficient algorithms and lower bounds for both versions
on three classes of polygons of increasing generality: convex, simple, and self-overlapping.
For convex polygons,
we solve the value version in $O(\log n)$ time and the reporting version in 
$O\!\left(h \log\left(1 + \frac{n}{h}\right)\right)$ time, where $h$ is the width of $P$ in direction $\vv{u}$. 
We prove matching lower bounds in the decision-tree model, showing that the reporting algorithm is input-sensitive optimal with respect to $h$.
For simple polygons, we present $O(n \log n)$-time, $O(n)$-space algorithms for both versions and prove an $\Omega(n)$ lower bound in the decision-tree model.
For self-overlapping polygons, we extend the approach for simple polygons to obtain $O(n \log n)$-time, $O(n)$-space algorithms for both versions, 
and we prove a matching $\Omega(n \log n)$ lower bound in the algebraic computation-tree model via a reduction from the $\delta$-closeness problem.

Our approach relies on a lattice-theoretic formulation of the problem.
We represent strip partitions as antichains of intervals in the 
Clarke--Cormack--Burkowski lattice, 
originally developed for minimal-interval semantics in information retrieval. 
Within this lattice framework, we design a dynamic programming algorithm 
that uses the lattice operations of meet and join.
To the best of our knowledge, this is the first geometric application of the Clarke--Cormack--Burkowski lattice.
\end{abstract}

\section{Introduction}\label{sec:introduction}
Partitioning a polygon into the minimum number of pieces 
has been studied extensively under shape-class constraints,
for example into convex, star-shaped, quadrilateral, or monotone pieces~\cite{abrahamsen_minimum_2024,chazelle_decomposing_1979,keil_decomposing_1985}. 
Much less is known for metric constraints that bound 
geometric measures, such as width, diameter, or perimeter.
In contrast to the classical shape-based setting, 
the problem becomes significantly harder under such metric bounds. 
Imposing a unit upper bound on diameter or perimeter already makes the minimum partition problem NP-hard for simple polygons without holes~\cite{abrahamsen_hardness_2024}.
Moreover, no polynomial-time algorithm is known for computing an optimal solution even when the input is a square or an equilateral triangle~\cite{abrahamsen_partitioning_2025}.



In ISAAC'25, Chung et al.~\cite{chung_minimum_2025} introduced a minimum partition problem with an upper-bounded width constraint, 
which limits the width of each partitioned piece in prescribed directions. 
Such a constraint is motivated by manufacturing and recycling, where materials must be cut or processed within width limits. 
For instance, Lan et al.~\cite{lan_utilization_2023} recycle wind turbine blades 
by first cutting the blades into panels small enough to fit the intake of crushing machines, and then crushing them into recyclates.
Moreover, wind turbine blades are typically made of glass-fiber composites with directionally aligned fibers, so the cutting effort can depend on the cut direction~\cite{hearle_structural_1967}.

From a theoretical perspective, the problem is related to 
Bang's conjecture, a long-standing open problem in convex geometry~\cite{bang_solution_1951}. 
Chung et al. resolved a partition analogue of this conjecture, showing that 
for any convex body, there exists a direction 
in which an optimal partition can be achieved using only parallel cuts. 

These applications and theoretical results motivate the study of minimum partitioning 
under an upper-bounded width constraint. 
Moreover, they suggest that it is natural to restrict attention to partitions
in which all cuts are parallel to a single direction, which 
we call the \emph{strip partition model}.


We consider three classes of input polygons: convex, simple, and self-overlapping, 
where a self-overlapping polygon generalizes a simple polygon by allowing the boundary to self-intersect.
A \emph{self-overlapping curve} was introduced by Shor and van Wyk~\cite{shor_detecting_1992} as the boundary of a topological disk immersed in $\mathbb{R}^2$. When this curve is polygonal, 
it serves as the boundary of a self-overlapping polygon. 
Unlike in the simple case, the boundary alone does not determine a unique interior or visibility relation. 
Following Shor and van Wyk, we define visibility with respect to a triangulation of the curve, so that partitions of the polygon are well defined. Further details on self-overlapping polygons are given in Section~\ref{sec:extend.overlapping}.

\begin{figure}[t!]
  \centering
  \includegraphics[width=0.9\textwidth]{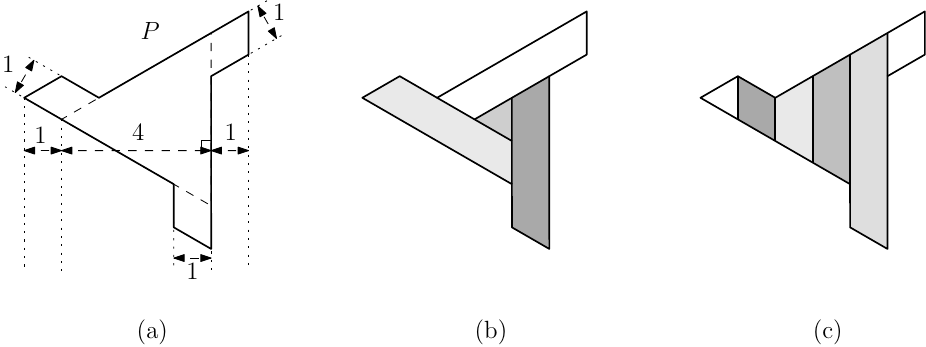}
  \caption{\small (a) The polygon P has a windmill shape with three arms extending from an equilateral triangle of height~4. (b) A minimum partition of $P$ under width and cut constraints, $W=U=\usetp$.
  (c) A minimum orthogonal strip partition of $P$ with $\vv{u}=(1,0), \vv{v} = (0,1)$.}
  \label{fig:strip.partition}
\end{figure}

In this work, we study orthogonal strip partitioning of polygons 
and present efficient algorithms with lower bounds for each class of input polygons.

\subparagraph{Orthogonal strip partition problem.}
We first recall the partitioning problem under width and cut constraints introduced by Chung et al.~\cite{chung_minimum_2025}. 
Let $P$ be a simple polygon, and let $\usetp = \{(\cos\theta, \sin\theta) \mid 0 \le \theta < \pi\}$ denote the set of all unit vectors in the plane. 
For $\vv{v} \in \usetp$ and 
a set $X \subset \mathbb{R}^2$, 
let $\dwidth{\vv{v}}{X}$ denote the width of $X$ in direction~$\vv{v}$, that is, the length of the projection of $X$ onto a line parallel to~$\vv{v}$. 
Given a subset $W \subseteq \usetp$, we say that $X$ satisfies the \emph{unit-width constraint}~$W$ if $\dwidth{\vv{v}}{X} \le 1$ for some $\vv{v} \in W$. 
We also impose a \emph{cut constraint} $U \subseteq \usetp$, meaning that every cut used to partition $P$ must be aligned with some direction in~$U$.
Here, a cut is a line segment $c \subseteq P$ 
whose relative interior lies in the interior of $P$.
Given a simple polygon $P$ and sets $W, U \subseteq \usetp$, 
the goal is to find a partition of $P$ into the minimum number of pieces such that every cut is aligned with a direction in $U$ and each piece satisfies the unit-width constraint~$W$.

The \emph{orthogonal strip partition problem} is the special case with $W = \{\vv{u}\}$ and $U = \{\vv{v}\}$ for two orthogonal unit vectors $\vv{u}, \vv{v} \in \usetp$. 
A \emph{strip} is a subpolygon $Q \subseteq P$ with $\dwidth{\vv{u}}{Q} \le 1$, and a strip partition divides $P$ into strips using only cuts parallel to~$\vv{v}$. 
Without loss of generality, we assume $\vv{u} = (1,0)$ and $\vv{v} = (0,1)$. 
Under this assumption, each cut is a vertical segment and every strip $Q$ has width at most~1 in the horizontal direction. 
In this orthogonal setting, 
we regard a cut as a maximal vertical segment contained in~$P$ whose relative interior lies in the interior of $P$.
See Figure~\ref{fig:strip.partition} for an illustration.

\subparagraph{Cut descriptors and lossless encoding.}
We define the \emph{descriptor} of a cut as the pair of boundary edges of $P$ on which its top and bottom endpoints lie, together with its $x$-coordinate. 
If a top or bottom endpoint lies on two incident edges, we choose the left edge, 
so the descriptor uniquely determines the position of the cut.
A partition of $P$ can be represented by the descriptors of its cuts. 
Our primary objective is to minimize the number of strips. In the reporting variant, we also compute the descriptors of all cuts in an optimal partition.

To store these descriptors succinctly, 
we fix a \emph{lossless encoding} scheme for strip partitions.
For each polygon $P$, the scheme assigns a codeword to each strip partition of $P$
so that the partition can be reconstructed from the codeword without loss. 
Thus each strip partition of $P$ is represented by a unique codeword.

The purpose of lossless encoding is to separate the combinatorial difficulty of finding an optimal partition from the cost of listing all cuts.
For an integer $k\ge1$,  
the optimal strip partition of $k\times 1$ rectangle
consists of $k-1$ cuts spaced one unit apart.
If we store one descriptor per cut, any reporting algorithm requires 
$\Omega(k)$ time to output the explicit list.
We therefore allow the algorithm to output a succinct codeword from which the full list 
of cut descriptors can be reconstructed without loss,
for example by encoding the $x$-coordinate of the leftmost cut together with the number of cuts.
This is called a \emph{run-length} encoding~\cite{bell_text_1990}.
Under this model, the cost is no longer dominated by explicit output size, but by distinguishing among combinatorially different optimal partitions.

\begin{definition}[Orthogonal Strip Partition Problems]
An orthogonal strip partition of a polygon $P \subset \mathbb{R}^2$ is a partition of $P$
into subpolygons (strips) using vertical cuts, such that 
each subpolygon has horizontal width at most~$1$.
We consider three input domains: 
convex, simple, and self-overlapping.
For $\mathcal{D} \in \{\text{convex}, \text{simple}\}$, an input instance 
consists of the vertices of $P$ listed in counterclockwise order along its boundary. 
For $\mathcal{D} = \text{self-overlapping}$, 
an instance consists of a triangulation of $P$.
For each domain $\mathcal{D}$, we consider two versions: 
\begin{itemize}\denseitems
    \item \emph{Value version:} return the minimum number of strips. 
    \item  \emph{Reporting version:} 
    with respect to the fixed lossless encoding scheme, 
    output any codeword representing an optimal strip partition of $P$.
\end{itemize}
\end{definition}

Let $\mathcal{D} \in \{\mathrm{conv}, \mathrm{simp}, \mathrm{self}\}$
denote the type of $P$, 
corresponding to convex, simple, and self-overlapping, respectively.
For a polygon $P$ in domain $\mathcal{D}$, let $\prob{V}{\mathcal{D}}(P)$ and $\prob{R}{\mathcal{D}}(P)$ 
be the value and reporting versions of the orthogonal strip partition problem, respectively.
We denote by $\nopt(P)$ the minimum number of strips in an optimal strip partition of $P$. 

\subparagraph{Related work.} 
The general version of the strip partition problem with
arbitrary width and cut constraints $W,U\subseteq \usetp$,
was first introduced by Chung et al.~\cite{chung_minimum_2025}. 
They studied structural properties in simple polygons, in particular the monotonicity of the minimum partition number under polygon containment. 
They proved a partition analogue of Bang's conjecture.
If $U$ contains every direction orthogonal to each $\vv{w} \in W$,
then any convex body admits 
an optimal partition obtained by equally spaced cuts parallel to some direction $\vv{v} \in U$. 
Hence an optimal solution exists in the strip partition model.

Most algorithmic work for the strip partition model 
focuses on monotone variants, including the convex case.
Liu~\cite{liu_partitioning_1985} gave an $O(n\log n)$-time algorithm for 
partitioning 
a rectilinear polygon with rectilinear holes 
into $x$-monotone pieces using horizontal cuts, where $n$ is the total number of vertices.
Lee et al.~\cite{lee_monotone_2025} 
improved this 
to linear time for hole-free polygons.
Lingas and Soltan~\cite{lingas_minimum_1998} 
studied minimum convex partitions of polygons with holes 
by cuts restricted to a set $U$ of directions; 
in particular, they give an $O(n\log n)$-time algorithm when $|U|=1$ (parallel cuts).

\subparagraph{Our results.}
We study orthogonal strip partitioning of convex, simple, 
and self-overlapping polygons, and give efficient algorithms with matching lower bounds in appropriate computational models, in particular the decision-tree and algebraic computation-tree models.

First, for a convex $n$-gon, the value version can be solved in $O(\log n)$ time and the reporting version in $O\bigl(h \log (1 + n/h)\bigr)$ time, where $h \ge 1$ is the horizontal width of the input polygon. 
As a function of $h$, this bound is $O(\log n)$ when $h=1$ and approaches $O(n)$ as $h$ increases.
We prove matching lower bounds in the decision-tree model, showing that any algorithm needs $\Omega(\log n)$ time for the value version and $\Omega\bigl(h \log (1 + n/h)\bigr)$ time for the reporting version. Thus the reporting algorithm is input-sensitive optimal with respect to~$h$.

Second, for simple polygons with $n$ vertices, we use a lattice-theoretic formulation of strip partitions. We represent strip partitions as antichains of intervals in the Clarke--Cormack--Burkowski lattice, originally developed for minimal-interval semantics in information retrieval, and design a dynamic programming algorithm on a trapezoidal decomposition whose recurrence is expressed using the lattice operations of meet and join on these antichains. This yields an $O(n \log n)$-time, $O(n)$-space algorithm that solves both the value and reporting versions. 
We also show that, unlike for convex polygons, any algorithm requires $\Omega(n)$ accesses to the coordinates of the input vertices in the worst case.
To the best of our knowledge, this is the first geometric application of the Clarke--Cormack--Burkowski lattice.

Finally, we extend this lattice-based framework to self-overlapping polygons given by a triangulation. We obtain an $O(n \log n)$-time, $O(n)$-space algorithm that solves both versions, and prove an $\Omega(n \log n)$ lower bound in the algebraic computation-tree model by a reduction from the $\delta$-closeness problem. This shows that our algorithm is optimal in that model.


\section{Preliminaries}\label{sec:preliminary}
Let $P$ be a polygon with $n$ vertices in the plane, given as a list of vertices in counterclockwise order along its boundary. 
A \emph{partition} of $P$ is a set of connected pieces with pairwise disjoint interiors whose union equals $P$. 
The \emph{cardinality} of a partition is the number of its pieces.

For a set $X \subseteq \mathbb{R}^2$, we denote by $\mybd{X}$ its boundary, by $\myint{X}$ its interior, and by $\mycl{X}$ its closure. 
We regard a polygon as the union of its interior and boundary; in particular, $\mycl{P} = P$, $\mybd{P}$ is the boundary of $P$, and $\myint{P}$ is its interior.

For a point $p \in \mathbb{R}^2$, let $\xcoord(p)$ and $\ycoord(p)$ denote its $x$- and $y$-coordinates, respectively. 
For any two points $p, q \in \mathbb{R}^2$, we use $pq$ to denote the line segment connecting $p$ and $q$, and write $|pq|$ for its length. 
We call $pq$ a \emph{cut} in $P$ if $p, q \in \mybd{P}$ and the interior points of $pq$ lie in $\myint{P}$. 
If $\ell$ is a vertical cut, we denote by $x(\ell)$ the $x$-coordinate of $\ell$.

Let $O$ be a totally ordered set.
A subset $I \subseteq O$ is an \emph{interval} of $O$ if for all $x,y,z \in O$ with $x \le z \le y$, the condition $x,y \in I$ implies $z \in I$. 
A \emph{closed interval} of $O$ is an interval of the form $[a,b] = \{x \in O \mid a \le x \le b\}$. 
We denote by $\cinterval_O$ the set of all closed intervals of $O$. 
If $O\subseteq \mathbb{R}$ with the induced order, then
for a closed interval $I = [a,b] \in \cinterval_O$, we define its length by $|I| \coloneqq b - a$.
The set $O$ is \emph{locally finite} if every closed interval $[a,b] \subseteq O$ contains only finitely many elements of $O$; 
for example, $\mathbb{Z}$ is locally finite, whereas $\mathbb{Q}$ and $\mathbb{R}$ are not.


We use the notation $[m]\coloneqq \{1,2,\ldots,m\}$ for a positive
integer $m$.  For a finite set $A$, we use $|A|$ to denote its cardinality.

\subparagraph{Computational model.} 
We work in the real-\texttt{RAM} model, which supports unit-cost arithmetic ($+,-,\times,\div$) and comparisons on real numbers. 
The $n$ vertices of the input polygon $P$ are given by real coordinates $(x_i,y_i)_{i=1}^n$. 
We allow a restricted use of floor function: 
for each input $x_i$, 
the integer part $\lfloor x_i \rfloor$ is also given and 
can be accessed in constant time, which increases the input size only by a constant factor.
We do not allow the floor function on arbitrary real values, since this would make the model unrealistically powerful~\cite{schonhage_power_1979}. 
The integer labels do not increase the computational power for the problems we consider.

Our algorithms take floor/ceiling of differences between $x$-coordinates of input vertices.
In our model, these are supported in constant time 
because the integer parts of all input $x$-coordinates are given. 
Let $(x_1, x_2,\ldots,x_n)$ be the $x$-coordinates of the vertices of $P$.
For each $i \in [n]$, 
let $f_i = \lfloor x_i \rfloor$ (the integer part) and $r_i = x_i - f_i$ (the fractional part).
Then, $\lfloor x_i - x_j\rfloor = f_i - f_j -\mathbf{1}\!\left[r_i<r_j\right]$, where $\mathbf{1}[\cdot]$ is 
the indicator function and its value is computed by a single comparison. 
For computing ceilings, 
we 
use the identity $\lceil x_i - x_j\rceil=-\lfloor x_j - x_i\rfloor$.

\section{Orthogonal strip partitioning of convex polygons}\label{sec:convex.case}
In this section, we consider the orthogonal strip partition problem on a convex polygon
$P$ with $n$ vertices, stored in counterclockwise order along its boundary. 
Since all cuts are required to be vertical, an optimal partition achieving $\nopt(P)$ is obtained by placing vertical cuts at unit distance, starting from the leftmost vertex of $P$. 
This also follows from the general result of Chung et al.~\cite{chung_minimum_2025} on optimal partitions of convex polygons under width and cut constraints.

\begin{corollary}[\cite{chung_minimum_2025}, Corollary~17]\label{cor:conv.min,partition}
Let $P$ be a convex polygon with $n$ vertices, and let $U, W \subseteq \usetp$ be sets of unit vectors such that every direction orthogonal to some $\vv{v} \in W$ 
is contained in $U$. 
Then, an optimal partition under width constraint $W$ and cut constraint $U$ is achieved by equally spaced cuts orthogonal to a direction $\vv{u} \in W$ that minimizes $\lceil \dwidth{\vv{u}}{P} \rceil$. Given such a direction $\vv{u}$, the partition can be computed in time $O\left(\dwidth{\vv{u}}{P} \log\left(1 + n/\dwidth{\vv{u}}{P}\right)\right)$.
\end{corollary}

By Corollary~\ref{cor:conv.min,partition}, if $h$ is the horizontal width of $P$, then
$\nopt(P) = \lceil h \rceil$, and an optimal partition is obtained by equally spaced
vertical cuts at unit distance.
The value $\nopt(P)$ can therefore be computed in $O(\log n)$ time by first computing
$h$ via a binary search on the list storing the vertices of $P$.

To report an optimal strip partition, we first consider the straightforward approach of
storing each cut descriptor explicitly as the pair of boundary edges spanned by the cut
together with its $x$-coordinate; we refer to this representation as the \emph{literal
encoding}.
A strip partition may, however, contain long sequences of consecutive cuts that span the
same edge pair and differ only in their $x$-coordinates.
For example, a sequence of $m$ cuts with $x$-coordinates $\{t, t+1, \ldots, t+m-1\}$ may
all span the same edge pair.
In this case it is more compact to use a \emph{run-length encoding}~\cite{bell_text_1990},
which stores the common edge pair together with the pair $(t,m)$, where $t$ is the
$x$-coordinate of the first cut and $m$ is the run length.
A literal encoding of a partition with $m$ cuts requires $O(m)$ space, whereas its
run-length encoding uses only $O(n)$ space.
Both encodings are lossless, and each cut can be reconstructed in $O(1)$ time.

The running-time bound in Corollary~\ref{cor:conv.min,partition} is obtained by using
the algorithm of Chung et al.~\cite{chung_approximating_2022}, originally developed in
a different context.
Their implementation uses a hybrid encoding of cuts that combines the literal and run-length encodings. 
When $h < n$, the cuts are stored in literal form; otherwise they are stored in run-length form. 
This hybrid strategy yields an $O\!\left(h \log\!\left(1 + \frac{n}{h}\right)\right)$-time algorithm that uses $O(n)$ space.
As a function of $h$, the running time is $\Theta(\log n)$ when $h = 1$ and increases monotonically, approaching $\Theta(n)$ as $h \to \infty$.

\begin{theorem}\label{thm:conv.algorithm}
Given a convex polygon $P$ with $n$ vertices stored in an array in counterclockwise order along its boundary, 
the problems $\prob{V}{conv}(P)$ and $\prob{R}{conv}(P)$ can be solved in $O(\log n)$ and $O\!\left(h \log\!\left(1 + \frac{n}{h}\right)\right)$ time, respectively, 
where $h$ is the horizontal width of $P$.
\end{theorem}

\subsection{Lower bound for the value problem}

Let $\mathbf{x} = (x_1, x_2, \ldots, x_n)$ be a circularly sorted list of $n$ distinct integers.
Consider the problem of computing $\max_i x_i$.
We reduce this problem to an instance of $\prob{V}{conv}(P)$ as follows.

Given $\mathbf{x}$, construct a convex polygon $P$ whose vertices are given by $p_i = (x_i, x_i^2)$ in the order of $\mathbf{x}$.
This yields a valid convex polygon, and 
its horizontal width is exactly $\max_i x_i - \min_i x_i$.
Therefore, any algorithm solving $\prob{V}{conv}(P)$ can be used to compute $\max_i x_i- \min_i x_i$ 
without additional queries.

Note that any comparison-based algorithm for computing the maximum of $n$ elements requires $\Omega(\log n)$ queries.
This is because there are $n$ possible positions for the maximum.
Moreover, any such algorithm induces a binary decision tree with at least $n$ leaves.
Thus, the same lower bound applies to $\prob{V}{conv}(P)$ in the decision-tree model.

\subsection{Lower bound for the reporting version}
Recall that each descriptor of a vertical cut in $P$ 
consists of a pair of boundary edges of $P$ spanned by the cut 
together with its $x$-coordinate. 
The problem $\prob{R}{conv}(P)$ requires, with respect to some lossless compression
$(E_P, D_P)$, outputting a codeword that encodes all cut descriptors of an optimal partition
of $P$. We now introduce a problem that reduces to an instance of $\prob{R}{conv}(P)$. 

\subparagraph{Interval-Index Coding problem.}
Let $m$ be a fixed positive integer.
The input is a sorted list $\mathbf{x} = (x_0, x_1, \ldots, x_n)$ of $n+1$ distinct real numbers
such that $0 = x_0 < x_1 < \cdots < x_n = m$.
For each $j \in \{1,\ldots,m\}$, let $z(j)$ be the unique index $i \in \{0,1,\ldots,n-1\}$ 
such that $x_i < j \le x_{i+1}$. 
The task is to report the mapping $j \mapsto z(j)$ using some lossless compression scheme $(E,D)$, 
that is to output a codeword $c$ such that $D(c) = (z(j))_{j=1}^{m}$. 

We reduce this problem to $\prob{R}{conv}(P)$ without additional queries.
Given $\mathbf{x}$, we construct a convex polygon $P$ whose vertices are
$p_i = (x_i, x_i^2)$ for all $0 \le i \le n$, in this order. Then
the horizontal width of $P$ is $m$.
For this instance, an optimal strip partition of $P$ uses exactly the vertical cuts at
$x = 1, 2, \ldots, m-1$. 
Since the upper hull of $P$ consists of a single edge,
the descriptor of the cut at $x=j$ is determined by 
the bottom edge on the lower hull that it intersects, which is precisely the edge $(p_{z(j)}, p_{z(j)+1})$.
Running an algorithm for $\prob{R}{conv}(P)$ on this polygon returns a codeword $c_P$
with respect to some fixed lossless compression $(E_P,D_P)$.
Decoding $c_P$ yields 
all cut descriptors of an optimal strip partition of $P$. 
Since each descriptor identifies the bottom edge $(p_{z(j)}, p_{z(j)+1})$ 
intersected by the cut at $x=j$, 
we can recover the entire sequence 
$(z(j))_{j=1}^m$
without additional queries.
 
Let $f$ be the bijection mapping each bottom edge $(p_{z(j)}, p_{z(j)+1})$ 
to its index $z(j)$.  
We just replace $(E_P, D_P)$ with $E: = E_P \circ f^{-1}$ and $D: = f \circ D_P$ 
without changing the returned codeword $c_P$.
Under this modified pair, the codeword $c_P$ is interpreted as an encoding of $(z(j))_{j=1}^m$, 
thereby completing the reduction.

\subparagraph{Search-problem formulation.}
We formalize the Interval-Index Coding problem as a search problem in the sense of
Lov\'{a}sz et al.~\cite{lovasz_search_1995}.
Let $I$ be 
the set of all valid input sentences 
\[
  I := \{\mathbf{x} = (x_0,\ldots,x_n) \in \mathbb{R}^{n+1}
        \mid 0 = x_0 < x_1 < \cdots < x_n = m \}.
\]
Let $W$ be the set of all nondecreasing integer sequences
\[
  W := \{\mathbf{c} = (c_1,\ldots,c_m) \in \{0,\ldots,n-1\}^m
        \mid c_1 \le \cdots \le c_m \}.
\]
For $\mathbf{x} \in I$ and $c \in W$, 
we define a relation $F \subseteq I \times W$ by 
\[
  (\mathbf{x}, \mathbf{c}) \in F
  \;\Longleftrightarrow\;
  \mathbf{c} = (z(j))_{j=1}^{m},
\]
where $(z(j))_{j=1}^{m}$ is the sequence determined from $\mathbf{x}$ in the
Interval-Index Coding problem.
Every valid input $\mathbf{x} \in I$ has 
at least one \emph{witness} $c\in W$ with $(\mathbf{x},c) \in F$.
Conversely, every sequence $\mathbf{c} \in W$
is realizable as $(z(j))_{j=1}^{m}$ for some valid input $\mathbf{x} \in I$.
The search problem defined by $F$ asks, given an input $\mathbf{x} \in I$, to output
a witness $c \in W$ such that $(\mathbf{x},c) \in F$.

In the deterministic decision-tree model, 
an algorithm for this search problem 
queries entries of the input sequence $\mathbf{x}$ and,
at a leaf, outputs a single witness $\mathbf{c} \in W$.
The deterministic decision-tree complexity of $F$ is
\[d(F):= \min\{\textsf{depth}(T) : T \text{ is a deterministic decision tree for } F\}. \]
Any decision tree $T$ for $F$ must have at least $|W|$ leaves, since different witnesses
in $W$ must appear as outputs at distinct leaves.
Thus,
$\textsf{depth}(T) \ge \lceil \log_2 |W| \rceil$.

Introducing a lossless compression scheme $(E,D)$ into the search problem does not
change the cardinality of the witness set $W$.
Because $(E,D)$ is lossless, different witnesses in $W$ must be mapped to
different codewords and can be recovered from them.
Thus there is a one-to-one correspondence between $W$ and the set of possible
codewords, so the search problem still has $|W|$ distinct witnesses.
Therefore, the number of distinct witnesses 
equals the number of ways to place 
$m$ indistinguishable balls into $n$ ordered bins (the classical stars-and-bars argument), 
that is $|W| = \binom{m+n-1}{n-1}$. 

\begin{lemma}\label{lem:log.binom}
    For positive integers $n,m$, 
    $\log_2\binom{m+n-1}{n-1} = \Omega(m\log(1+n/m))$.
\end{lemma}
\begin{proof}
The inequality is trivial for $n=1$, since $1 \ge m \ln (1+\frac{1}{m})$ for all $m>0$.
So we assume $n \ge 2$.
By symmetry, 
\[
\binom{m+n-1}{n-1}
= \binom{m+n-1}{m}
= \prod_{i=1}^m \frac{(n-1)+i}{i}
= \prod_{i=1}^m \left(1 + \frac{n-1}{i}\right).
\]
Since $1 \le i \le m$ and $n\ge 2$, we have
\[
1+\frac{n-1}{i}
\;\ge\;
1+\frac{n-1}{m}
\;\ge\;
1+\frac{n}{2m}
\;\ge\;
\sqrt{1+\frac{n}{m}},
\]
where the last inequality follows from $(1+x/2)^2 \ge 1+x$ for all $x\ge 0$.
Thus,
\[
\binom{m+n-1}{n-1}
= \binom{m+n-1}{m}
\;\ge\; \left(1+\frac{n-1}{m}\right)^m \;\ge\; \left(1+\frac{n}{m}\right)^{m/2}.\]
Taking base-2 logarithms of both sides completes the proof.
\end{proof}

By Lemma~\ref{lem:log.binom}, we have $d(F) = \Omega\bigl(m \log(1 + n/m)\bigr)$.
This lower bound carries over to $\prob{R}{conv}(P)$ in the decision–tree model, 
yielding a lower bound of $\Omega\bigl(h \log (1+n/h)\bigr)$ for any convex $n$-gon with horizontal width $h\ge 1$.

\begin{theorem}\label{thm:conv.alg.optimality}
Let $P$ be a convex polygon with $n$ vertices stored in an array in order along its boundary,
and let $h\ge 1$ be a horizontal width of $P$.
In the algebraic decision-tree model, 
any algorithm for $\prob{V}{conv}(P)$ requires
$\Omega(\log n)$ time, and any algorithm for $\prob{R}{conv}(P)$
requires $\Omega\bigl(h \log (1 + n/h)\bigr)$ time in the worst case.
\end{theorem}

Combining Theorem~\ref{thm:conv.algorithm} and~\ref{thm:conv.alg.optimality}, 
the algorithms for convex polygons 
are optimal in the (algebraic) decision-tree model.
Note that the horizontal width $h$ of $P$ is an input-sensitive parameter of $P$
that can be computed in $O(\log n)$ time. 
Thus, the algorithm for $\prob{R}{conv}(P)$ is input-sensitive optimal with respect to $h$ in the decision-tree model.

\begin{figure}[b!]
  \centering
  \includegraphics[width=\textwidth]{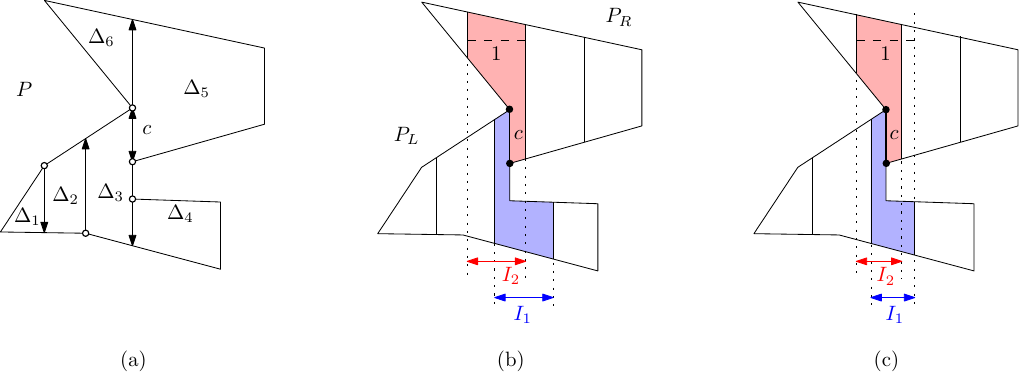}
  \caption{\small (a) Vertical trapezoidal decomposition of $P$ into $\td = \{\Delta_1,\ldots,\Delta_6\}$, and decomposition of $P$ into $P_L$ and $P_R$ by a vertical cut $c$ of $\td$.  (b-c) Merging optimal strip partitions of $P_L$ and $P_R$ incident to $c$ may or may not require an additional cut, depending on whether $|I_1 \cup I_2| \le 1$ or $> 1$.}
  \label{fig:divide.polygon}
\end{figure}

\section{Lattice formulation for orthogonal strip partitions of simple polygons}\label{sec:simple.alg}
In this section, $P$ is a simple polygon with $n$ vertices, stored in an array in counterclockwise order along its boundary.
We give an $O(n \log n)$-time, $O(n)$-space algorithm for both $\prob{V}{\mathrm{simp}}(P)$ and $\prob{R}{\mathrm{simp}}(P)$.

We first compute a vertical trapezoidal decomposition of $P$ and its dual graph in $O(n)$ time using the algorithm of Chazelle~\cite{chazelle_triangulating_1991}; see Figure~\ref{fig:divide.polygon}(a).
Let $\td = \{\Delta_1,\Delta_2,\ldots,\Delta_m\}$ be the set of trapezoids and let $G = (\td,E)$ be its dual graph. 
Then $m = O(n)$ and $G$ is a tree.

Let $c$ be a vertical cut in the trapezoidal decomposition.
Cutting $P$ along $c$ yields two subpolygons $P_L,P_R \subset P$, where
$P_L$ is the part of $P$ lying to the left of $c$ and $P_R$ is the part
lying to the right.
Consider optimal strip partitions of $P_L$ and $P_R$ in which every strip has
horizontal width at most~$1$.
Taking the union of these two partitions, 
with $c$ as a common boundary, yields a feasible strip partition of $P$,
implying $\nopt(P) \le \nopt(P_L) + \nopt(P_R)$.
Conversely, given an optimal partition of $P$, 
restricting it to $P_L$ and $P_R$ yields feasible strip partitions of these subpolygons.
Thus, $\nopt(P_L) + \nopt(P_R) -1 \le \nopt(P) \le  \nopt(P_L) + \nopt(P_R)$,  
and we need to determine whether the lower or upper bound is achieved in a given instance.

To analyze this, observe that in any strip partition of $P_L$, 
there is exactly one strip incident to $c$.
For $P_L$, let $\mathcal{I}_c(P_L)$ denote the set of $x$-intervals obtained as follows:
for each optimal strip partition of $P_L$, take the strip incident to $c$ and project it
onto the $x$-axis, and include the resulting interval in $\mathcal{I}_c(P_L)$.
We define $\mathcal{I}_c(P_R)$ for $P_R$ analogously.
Every interval in $\mathcal{I}_c(P_L)$ and in $\mathcal{I}_c(P_R)$ 
contains $x(c)$, 
since its corresponding strip is incident to $c$.

When gluing optimal partitions of $P_L$ and $P_R$ along $c$, the total number of
strips decreases by one only when the two strips incident to $c$ can be merged
into a single strip without violating the width constraint.
In terms of $x$-intervals, this is equivalent to checking whether there exist
intervals $I_1 \in \mathcal{I}_c(P_L)$ and $I_2 \in \mathcal{I}_c(P_R)$ such that
$\lvert I_1 \cup I_2 \rvert \le 1$. See Figure~\ref{fig:divide.polygon}.

\begin{observation}\label{lem:determine.simple.opt}
We have $\nopt(P) = \nopt(P_L) + \nopt(P_R) - 1$ if and only if 
there exist intervals $I_1 \in \mathcal{I}_c(P_L)$ and $I_2\in \mathcal{I}_{c}(P_R)$ such that 
$|I_1 \cup I_2|\le1$. 
\end{observation}

Moreover, 
it suffices to consider only inclusion-minimal intervals in each family $\mathcal{I}_c(P_L)$ and $\mathcal{I}_c(P_R)$.
If $(I_1,I_2) \in \mathcal{I}_c(P_L)\times \mathcal{I}_c(P_R)$ satisfies 
$|I_1 \cup I_2| \le 1$ and there exists $I_1' \in \mathcal{I}_c(P_L)$ with $I_1' \subsetneq I_1$,
then $|I_1' \cup I_2| \le 1$ also holds; by symmetry, the same argument applies to $\mathcal{I}_c(P_R)$.
Thus, we may restrict attention to inclusion-minimal intervals in
$\mathcal{I}_c(P_L)$ and $\mathcal{I}_c(P_R)$.

\begin{observation}\label{obs:minimal-intervals}
To decide whether there exist
$I_1 \in \mathcal{I}_c(P_L)$ and $I_2 \in \mathcal{I}_c(P_R)$ with 
$|I_1 \cup I_2| \le 1$, it suffices to consider only the inclusion-minimal
intervals in $\mathcal{I}_c(P_L)$ and $\mathcal{I}_c(P_R)$.
\end{observation}

\subparagraph{Algorithm overview and lattice formulation.}
Our algorithm performs a bottom-up dynamic program over the vertical trapezoidal decomposition of $P$ and its dual tree $G$, rooted at an arbitrary node.
For a node $u$ of $G$, let $Q$ be the subpolygon formed by the union of trapezoids in the subtree rooted at $u$.
If $u$ is not the root, we designate a vertical edge $c$ of $Q$ that serves as the interface to its parent.
The associated subproblem for $(Q,c)$ 
is to determine the minimum number $\nopt(Q)$, 
together with the possible $x$-intervals of strips incident to $c$ in optimal partitions of $Q$.

Observation~\ref{lem:determine.simple.opt} shows that, when two subpolygons are glued across $c$, the optimum of the combined region can be determined from the optimum of each subpolygon and the $x$-intervals of the strips incident to $c$ in optimal partitions of the two subpolygons.
Moreover, by Observation~\ref{obs:minimal-intervals}, it suffices to retain only the inclusion-minimal such intervals.
Thus, for each subproblem $(Q,c)$, we maintain the optimal value $\nopt(Q)$ together with the inclusion-minimal intervals in $\mathcal{I}_c(Q)$.


We initialize these DP states at leaf trapezoids and combine them bottom-up at internal nodes, 
eventually reaching the root and obtaining $\nopt(P)$.
The crucial observation is that the update from child subproblems to a parent subproblem 
can be carried out entirely on these interval families.
These interval families form an order-theoretic structure, 
namely the Clarke--Cormack--Burkowski (CCB) lattice.
This lattice comes with two basic operations, meet($\wedge$) and join($\vee$).
As we show next, the two basic operations naturally describe the two update situations of the dynamic program:
gluing two incident strips into one, or introducing a new cut when such a gluing would violate the unit-width constraint.


\subsection{Antichain completion and the Clarke–Cormack–Burkowski lattice} 
Let $(X,\le)$ be a partially ordered set (poset). 
A subset $A \subseteq X$ is an \emph{antichain} if no two distinct elements of $A$ are comparable, that is, for all $a,b \in A$ with $a \ne b$, we have neither $a \le b$ nor $b \le a$. 
We denote by $\ac{X}$ the set of all antichains of $X$.

Boldi and Vigna~\cite{boldi_lattice_2018} 
order antichains by their lower sets. 
For $\mathcal{A} \in \ac{X}$, its \emph{lower set} is
$\downarrow\!(\mathcal{A}) := 
\{ x \in X \mid \exists y \in \mathcal{A} \text{ with } x\le y \}$.
They define a partial order $\preceq$ 
on $\ac{X}$ by
$\mathcal{A} \preceq \mathcal{B}
  \Longleftrightarrow
  \downarrow\!(\mathcal{A}) \subseteq \downarrow\!(\mathcal{B})$. 
  The poset $(\ac{X}, \preceq)$ is called the \emph{antichain completion} of $(X,\le)$.

Now let $O$ be a totally ordered set, and let $\cinterval_O$ be the set of all closed intervals of $O$. 
We order $\cinterval_O$ by reverse inclusion, so $(\cinterval_O,\supseteq)$ is a poset. 
We write $\ccb_O \coloneqq \ac{\cinterval_O}$ and view $(\ccb_O,\preceq)$ as the antichain completion of $(\cinterval_O,\supseteq)$.
For our purposes, it suffices to keep in mind the following description of $\ccb_O$ and $\preceq$.
Each element $\mathcal{A} \in \ccb_O$ is a family of intervals in which no interval properly contains another. 
For $\mathcal{A}, \mathcal{B} \in \ccb_O$, the relation $\mathcal{A} \preceq \mathcal{B}$ means that every interval in $\mathcal{A}$ contains at least one interval in~$\mathcal{B}$.

\begin{figure}[t!]
  \centering
  \includegraphics[width=\textwidth]{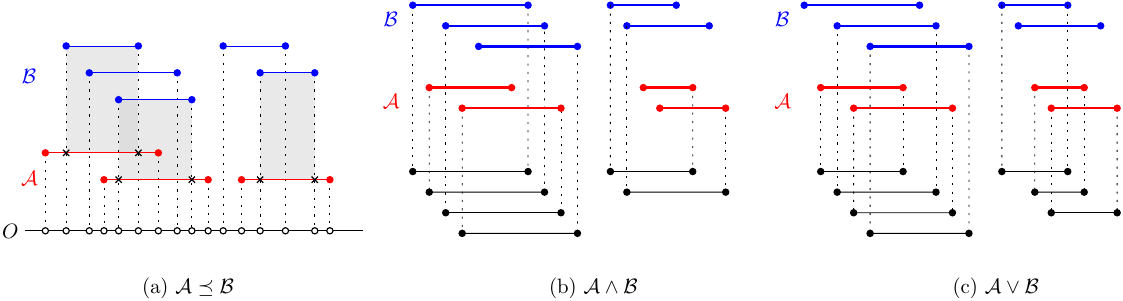}
  \caption{\small Let $(\ccb_O,\preceq)$ be the antichain completion of $(\cinterval_O,\supseteq)$. The figure illustrates the partial order, meet, and join operations for two antichains $\mathcal{A}, \mathcal{B} \in \ccb_O$: 
(a) $\mathcal{A} \preceq \mathcal{B}$; 
(b) $\mathcal{A} \wedge \mathcal{B}$; 
(c) $\mathcal{A} \vee \mathcal{B}$.}
  \label{fig:ccb.lattice}
\end{figure}

\medskip

\subparagraph{Clarke--Cormack--Burkowski (CCB) lattice.}
A poset is a \emph{lattice} if every pair of elements has a unique least upper bound (join) and 
a unique greatest lower bound (meet). 
For a locally finite, totally ordered set $O$, Boldi and Vigna~\cite{boldi_lattice_2018} show that 
the antichain completion 
$\ccb_O$ forms a complete lattice. 
This structure, known as the \emph{Clarke--Cormack--Burkowski lattice}, is named after the work of Clarke, Cormack, and Burkowski on minimal-interval semantics in information retrieval~\cite{clarke_algebra_1995}.

Let $O$ be a locally finite, totally ordered set, 
so $\ccb_{O}$ is a CCB lattice. 
Each $\mathcal{A} \in \ccb_{O}$ is a family of intervals of $O$ in which no interval properly contains another. 
For a family $\mathcal{F}$ of intervals, let
$\min_{\subseteq}(\mathcal{F})
  \coloneqq
  \{ I \in \mathcal{F} \mid \nexists J \in \mathcal{F} \text{ with } J \subsetneq I \}$
be the set of inclusion-minimal intervals in $\mathcal{F}$. 
The join and meet of $\mathcal{A},\mathcal{B} \in \ccb_{O}$ are defined, respectively, as
\[
  \mathcal{A} \vee \mathcal{B}
  \coloneqq \min_{\subseteq}(\mathcal{A} \cup \mathcal{B})
  \quad\text{and}\quad
  \mathcal{A} \wedge \mathcal{B}
  \coloneqq \min_{\subseteq}
     \bigl\{
       [\min\{\ell,\ell'\},\,\max\{r,r'\}]
       \mid [\ell,r]\in\mathcal{A},\ [\ell',r']\in\mathcal{B}
     \bigr\}.
\]
See Figure~\ref{fig:ccb.lattice} for an illustration of the partial order, meet, and join of two antichains.

Boldi and Vigna also show that $\ccb_{O}$ is a \emph{completely distributive} lattice: 
arbitrary meets distribute over arbitrary joins in this lattice. 
In particular, for any family $(\mathcal{A}_i)_{i\in I} \subseteq \ccb_{O}$ and any $\mathcal{B} \in \ccb_{O}$,
we have $\Bigl(\bigvee_{i\in I} \mathcal{A}_i\Bigr) \wedge \mathcal{B}
  \;=\;
  \bigvee_{i\in I} \bigl(\mathcal{A}_i \wedge \mathcal{B}\bigr)$. 

\subparagraph{Lattice-based representation of strip partitions.}
For a vertical cut $c$, the sets 
$\min_{\subseteq}(\mathcal{I}_c(P_L))$ and $\min_{\subseteq}(\mathcal{I}_c(P_R))$
are antichains of intervals, since no interval in each set properly contains
another. By Observation~\ref{obs:minimal-intervals}, these antichains suffice to
decide whether the strips incident to $c$ can be merged.
In fact, the equality 
$\nopt(P) = \nopt(P_L) + \nopt(P_R) - 1$ holds if and only if
the meet
$\min_{\subseteq}(\mathcal{I}_c(P_L)) \;\wedge\;
\min_{\subseteq}(\mathcal{I}_c(P_R))$ contains an interval of length at most~$1$.
Thus,
in our dynamic program, 
we only need to maintain, for each subpolygon $Q$ with vertical edge~$c$, the antichain 
$\iant(Q,c) \coloneqq \min_{\subseteq}\bigl(\mathcal{I}_c(Q)\bigr)$.
In other words, $\iant(Q,c)$ 
is the set of inclusion-minimal $x$-intervals of strips incident to $c$ across all optimal strip partitions of $Q$.

At first glance, these antichains appear to be families of intervals of $\mathbb{R}$,
which is not locally finite.  
However, as we will show, once we fix some $\varepsilon \ge 0$,
every interval that arises in our algorithm has endpoints of the form
$x(p) + m$ or $x(p) + m \pm \varepsilon$, where $p$ is a vertex of $P$
and $m \in \mathbb{Z}$.
Thus all endpoints lie in
\[
  \mathbb{Z}(P,\varepsilon)
  := \{\,x(p) + m - \varepsilon,\; x(p) + m,\; x(p) + m + \varepsilon 
        \mid p \text{ is a vertex of } P,\; m \in \mathbb{Z} \,\},
\]
a locally finite, totally ordered set.
Consequently, every antichain maintained by our dynamic program lies in the 
CCB lattice $\ccb_{\mathbb{Z}(P,\varepsilon)}$.
This lattice is completely distributive, a property that we use  
in the lower-bound proof in Section~\ref{subsec:lower.self}.

\begin{figure}[t!]
  \centering
  \includegraphics[width=0.55\textwidth]{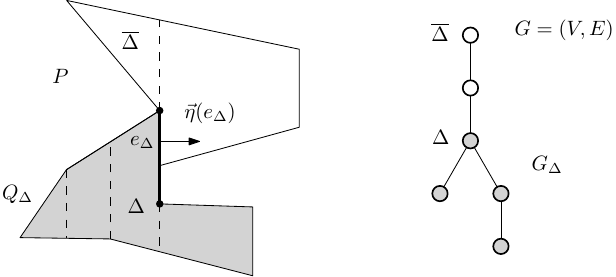}
  \caption{\small Subproblem structure $(Q_\Delta, e_\Delta, \normal(e_\Delta))$ is defined for every non-root node $\Delta$ of $G = (V,E)$.
}
  \label{fig:sub.structure}
\end{figure}


\subsection{Subproblem structures in the refined dual tree}\label{subsec.dynamic.approach}
We now describe the subproblem structure underlying our dynamic program for simple polygons,
formulated in terms of antichains of intervals.

We associate each subproblem with a pair $(Q,e)$, where $Q$ is a union of trapezoids from the vertical trapezoidal decomposition $\td$ and $e$ is a vertical boundary edge of $Q$. 
(Recall that $\mathcal{I}_e(Q)$ denotes the family of $x$-intervals of all strips incident to edge $e$ across all optimal strip partitions of $Q$.)
We compress this family into an antichain by keeping only inclusion-minimal intervals:
$\iant(Q,e) \coloneqq \min_{\subseteq}\bigl(\mathcal{I}_e(Q)\bigr)$.
Every interval in $\iant(Q,e)$ has length at most~$1$, and each such interval includes the common $x$-coordinate $x(e)$.

Let $G = (\td, E)$ be the dual graph of the trapezoidal decomposition $\td$, and fix an arbitrary leaf $\overline{\Delta} \in \td$ as the root. 
For each node $\Delta \in \td$, let $G_\Delta$ be the subtree of $G$ rooted at $\Delta$ and let $Q_\Delta \subseteq P$ be the subpolygon obtained as the union of the trapezoids in $G_\Delta$. 
A vertical cut of the trapezoidal decomposition that lies on $\mybd{Q_\Delta}$ is called a \emph{vertical interface} of $Q_\Delta$. 

Because $G$ is a tree, 
all the vertical interfaces of $Q_\Delta$ lie on a single vertical edge of $Q_\Delta$
for $\Delta\neq \overline{\Delta}$.
Let $e_\Delta$ denote this vertical edge.
Note that $e_\Delta$ is contained in $\mybd{\Delta}$.
We define the \emph{normal direction} of $e_\Delta$, denoted $\normal(e_\Delta)$, to be \textsf{right} if $e_\Delta$ lies on the right side of $\Delta$, and \textsf{left} if it lies on the left side. 
We refer to the triple $(Q_\Delta, e_\Delta, \normal(e_\Delta))$ as the \emph{subproblem structure} associated with~$\Delta$. See Figure~\ref{fig:sub.structure}.


\subparagraph{Refining $G$ into a binary tree $\widetilde G$.}
We transform the dual graph $G$ of $\td$ 
into a rooted binary tree $\widetilde G$ 
by locally refining each star centered at a trapezoid. 
Let $\Delta\in \td$ be a trapezoid in $\td$
with subproblem structure $(Q_\Delta,e_\Delta,\normal(e_\Delta))$, and 
assume without loss of generality that $\normal(e_\Delta)=\mathsf{right}$.
Suppose $\Delta$ has $p$ children incident to its left side and $q$ children on its right.
Let $(\Delta_i^L,e_i^L, \normal(e_i^L))_{i=1}^p$ (resp. $(\Delta_j^R,e_j^R, \normal(e_j^R))_{j=1}^q$) be the subproblem structures associated with its left (resp. right) children.
By construction, each $\normal(e_i^L) = \mathsf{right}$ and each $\normal(e_i^R) = \mathsf{left}$.

First, we construct a binary tree $T_L$ whose leaves are the left children of $\Delta$.
If $p = 0$, we do not create $T_L$;
if $p = 1$, then $T_L$ is a single node $\Delta_1^L$.
For $p \ge 2$, we
create $(p-1)$ new internal nodes
$u_2,\ldots,u_p$ and set each $u_i$ as the parent of 
$\Delta_{i}^L$ and $u_{i-1}$ for each $i=2,\ldots,p$, 
with $u_1 := \Delta_1^L$.
Symmetrically, we construct a binary tree $T_R$ for the right children of $\Delta$.

Next, we remove all edges from $\Delta$ to its children in $G$. 
If $T_L$ exists (i.e., $p\ge 1$), 
we attach the root of $T_L$ as the sole child of $\Delta$.
If $T_R$ exists, 
let $w$ be a new internal node whose children are the root of $T_R$ and $\Delta$; otherwise, let $w := \Delta$.
If $\Delta$ has a parent $z$ in $G$, 
we replace the edge between $z$ and $\Delta$ 
by an edge between $z$ and $w$.
The case $\normal(e_\Delta) = \mathsf{left}$ is handled symmetrically.
See Figure~\ref{fig:refine.binary.tree} for an illustration.

Processing all trapezoids in postorder 
replaces every star subgraph in $G$ 
with a binary tree.
At each step, only the adjacency between $\Delta$ and its children
changes; the subtrees rooted at those children remain intact.
This refinement adds only $O(n)$ new internal nodes in total, and 
$\widetilde G$ can be constructed in $O(n)$ time.
We refer to the original nodes of $G$ as \emph{trapezoid nodes} and
the new internal nodes as \emph{bridge nodes}.
By construction, every trapezoid node in $\widetilde G$ 
has at most one child, and 
every bridge node has exactly two children.

\begin{figure}[t!]
  \centering
  \includegraphics[width=0.8\textwidth]{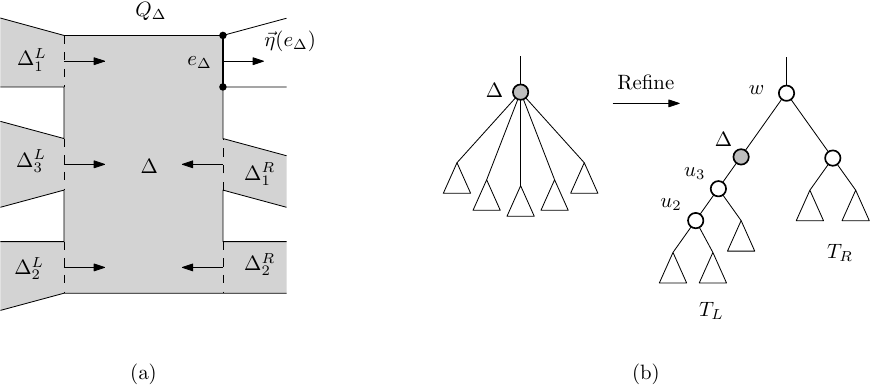}
  \caption{\small (a) The trapezoid $\Delta$ has three left children and two right children in the subtree $G_\Delta$.
(b) Refining the local star centered at $\Delta$ into a binary tree.}
  \label{fig:refine.binary.tree}
\end{figure}

\subparagraph{Subproblem structures in $\widetilde G$.}
We lift the subproblem structures defined for nodes of $G$
to the nodes of the refined tree $\widetilde G$.
For each non-root node $u$, we associate a triple
$(Q_u, e_u, \normal(e_u))$, where
$Q_u$ is a weakly simple subpolygon of $P$,
$e_u$ is a vertical edge on $\mybd{Q_u}$, and
$\normal(e_u)\in\{\mathsf{left},\mathsf{right}\}$ is
the outward normal direction of $e_u$ with respect to $Q_u$.
If $u$ is a trapezoid node corresponding to a trapezoid $\Delta\in\td$,
we inherit its subproblem structure by setting 
$(Q_u,e_u,\normal(e_u)) := (Q_\Delta,e_\Delta,\normal(e_\Delta))$.

Suppose that $u$ is a bridge node with children $v$ and $w$.  
Assume that $(Q_{v},e_{v},\normal(e_v))$ and 
$(Q_{w},e_{w},\normal(e_w))$ are already defined. 
Let $s$ be the minimal vertical segment spanning both $e_v$ and $e_w$, 
and define $Q_u := Q_v \cup s \cup Q_w$.
If $\normal(e_v) = \normal(e_w)$, then $u$ 
is a \emph{same-side bridge}.
We set
$e_u := s$ and $\normal(e_u) := \normal(e_v)$.
In this case, $Q_u$ has boundary edges that touch each other
without crossing, so $Q_u$ is only weakly simple.

If $\normal(e_v) \neq \normal(e_w)$, 
then $u$ is a
\emph{cross-side bridge}. 
In this case, $Q_u$ is always simple and 
one of $e_v, e_w$ properly contains the other.
If $e_v \subsetneq e_w$, 
we set $e_u := e_w \setminus e_v$ and $\normal(e_u) := \normal(e_w)$;
otherwise, $e_u := e_v \setminus e_w$ and $\normal(e_u) := \normal(e_v)$.
Figure~\ref{fig:bridge.node.cap}(a-b) illustrates the subproblem structures at bridge nodes.

For the root node $\rho$ of $\widetilde G$, 
which corresponds to the leaf $\overline{\Delta}$ of $G$,
the region $Q_\rho = P$ has no vertical interface on its boundary. 
Then, we cannot directly define a subproblem structure for $\rho$. 
Let $e$ be the vertical side of $\overline{\Delta}$ that is not incident to any other trapezoid of~$\td$.
Note that $e$ may be a single vertex if $\overline{\Delta}$ is a triangle.
In this degenerate case, 
we conceptually replace $e$ with a very short vertical segment, 
obtained by extending it slightly in both vertical directions.
This extension preserves the $x$-interval of 
$\overline{\Delta}$, keeps $P$ simple, and does not affect any strip partition of $P$. 
We set $e_\rho := e$ and define $\normal(e_\rho)$ as $\mathsf{left}$ 
if $e$ lies on the left side of $\overline{\Delta}$, or 
$\mathsf{right}$ if $e$ lies on the right side. 
This yields the subproblem structure $(Q_\rho, e_\rho, \normal(e_\rho))$ for the root node $\rho$.
See Figure~\ref{fig:bridge.node.cap}(c).

\begin{lemma}\label{lem:subproblem.structure}
For every node $u$ of $\widetilde G$, the subproblem structure
$(Q_u,e_u,\normal(e_u))$ is well-defined: 
$Q_u$ is a weakly simple subpolygon of $P$, $e_u$ is 
a vertical edge of $Q_u$ with positive length, 
and $\normal(e_u)$ is the outward normal direction of $e_u$ with respect to $Q_u$.
\end{lemma}
\begin{proof}
We proceed by induction on $\widetilde G$.
For a trapezoid node $u$, the triple $(Q_u,e_u,\normal(e_u))$ is inherited from the
original dual tree $G$, where it is already well-defined.
At the root node $\rho$, we have $Q_\rho = P$, the extended edge $e_\rho$ has positive length, and 
$\normal(e_\rho)$ is taken to be the outward normal of
$e_\rho$ with respect to $P$.

Now suppose that $u$ is a bridge node with children $v$ and $w$.
By construction, $e_v$ and $e_w$ are vertical segments on
$\mybd Q_v$ and $\mybd Q_w$ lying on a common vertical line.
If $\normal(e_v) = \normal(e_w)$ (same-side bridge), then $Q_v$ and $Q_w$ lie on
the same side of this line.
Let $e_u$ be the minimal vertical segment containing $e_v \cup e_w$, and define
$Q_u := Q_v \cup e_u \cup Q_w$ and
$\normal(e_u) := \normal(e_v) = \normal(e_w)$.
Then $Q_u$ is a weakly simple subpolygon of $P$, $e_u$ is a vertical boundary edge of
positive length, and $\normal(e_u)$ is the outward normal of $e_u$ with respect to $Q_u$ 
(away from the interiors of $Q_v$ and $Q_w$).

If $\normal(e_v) \neq \normal(e_w)$ (cross-side bridge), 
then $Q_v$ and $Q_w$ lie on
opposite sides of the common vertical line, and one of $e_v,e_w$ properly contains the other; 
without loss of generality, assume $e_w \supsetneq e_v$.
Along $e_v$, $Q_v$ and $Q_w$ are adjacent from opposite sides. 
Thus $Q_u := Q_v \cup Q_w$ is a simple polygon, 
and $e_u := e_w \setminus e_v$ is a vertical segment on $\mybd Q_u$. 
Since $u$ is not the root, 
$e_u$ contains a vertical interface of the trapezoidal decomposition, so $|e_u| > 0$. 
Thus, $Q_u$ is a simple subpolygon of $P$, $e_u$ is a vertical boundary edge of positive length, and $\normal(e_u) := \normal(e_w)$ is 
the outward normal of $e_u$ with respect to $Q_u$.
\end{proof}

\begin{figure}[t!]
  \centering
  \includegraphics[width=0.9\textwidth]{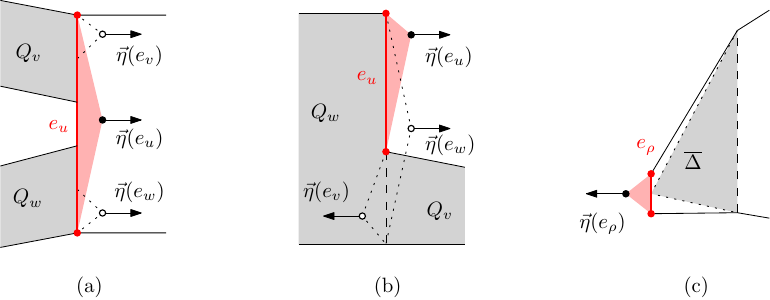}
  \caption{\small The subproblem structures $(Q_u, e_u, \normal(e_u))$ in $\widetilde G$ when $u$ is
(a) a same-side bridge, (b) a cross-side bridge, and (c) the root; attaching the cap $T_\rho(\varepsilon)$ along $e_\rho$ yields the simple polygon $Q_\rho^\varepsilon$.}
  \label{fig:bridge.node.cap}
\end{figure}

\subparagraph{Reducing weakly simple subpolygons to simple ones.}
By Lemma~\ref{lem:subproblem.structure}, every node $u$ of $\widetilde G$
has an associated subproblem structure $(Q_u, e_u, \normal(e_u))$.
If $u$ is a trapezoid node or a cross-side bridge node, then $Q_u$ is
simple. 
If $u$ is a same-side bridge node, however, 
$Q_u$ is only
weakly simple. 
To handle all subproblems in the simple polygon setting, we
convert $Q_u$ into a simple polygon
by attaching a thin triangular cap along $e_u$ on the outside of $Q_u$.

Let $q_u$ be any point on the edge $e_u$, excluding its endpoints. 
For sufficiently small $\varepsilon >0$, 
let $q_u(\varepsilon)$ be the point obtained by shifting 
$q_u$ by $\varepsilon$ units
in the direction of $\normal(e_u)$. 
Let $T_u(\varepsilon)$ be the triangle with base $e_u$ and apex $q_u(\varepsilon)$.

\begin{lemma}\label{lem:subproblem-cap}
For each node $u$ of $\widetilde G$, there exists
$\varepsilon_0>0$ such that $Q_u \cup T_u(\varepsilon)$ is a simple polygon 
for all $0 < \varepsilon \le \varepsilon_0$.
\end{lemma}
\begin{proof}
Without loss of generality, assume $\normal(e_u) =\texttt{right}$.
If $u$ is a same-side bridge node, then $Q_u$ is weakly simple.
Along the vertical edge $e_u$, 
several simple subpolygons of $P$ are attached on the left side of $e_u$, 
and a sufficiently small neighborhood to the right of $e_u$ is disjoint from $Q_u$.
Then $e_u$ is a vertical boundary edge of $Q_u$ with outward normal $\normal(e_u)$. 
Choose a point $p$ in the relative interior of $e_u$.
There exists $\varepsilon_0>0$ such that, for all
$0<\varepsilon\le\varepsilon_0$, the triangular cap $T_u(\varepsilon)$ with base
$e_u$ and apex at $p + \varepsilon \normal(e_u)$ has interior disjoint from $Q_u$
and meets $Q_u$ only along its base. For such $\varepsilon$, the union
$Q_u \cup T_u(\varepsilon)$ is a simple polygon.

If $u$ is a cross-side bridge node, then $Q_u$ is already a simple polygon and
$e_u$ is a vertical boundary edge of $Q_u$ with outward normal $\normal(e_u)$.
For any point $p$ in the relative
interior of $e_u$, the same local argument
shows that there exists $\varepsilon_0>0$ such that
$\myint{T_u(\varepsilon)} \cap \myint{Q_u} = \emptyset$ 
and $T_u(\varepsilon)$ meets $Q_u$ along its base for all
$0<\varepsilon\le\varepsilon_0$. Hence $Q_u \cup T_u(\varepsilon)$ is also simple.
\end{proof}

Let $(x_i)_{i=1}^n$ be the $x$-coordinates of the vertices of $P$, and let
$\mathbb{Z}(P,0) := \{\, x_i + m \mid i \in [n],\, m \in \mathbb{Z} \,\}$.
Let $\delta>0$ be the minimum positive distance between two distinct points of $\mathbb{Z}(P,0)$.
From Lemma~\ref{lem:subproblem-cap},
for each node $u$, there exists $\varepsilon_0(u)>0$
such that $Q_u \cup T_u(\varepsilon)$ is simple for all $0<\varepsilon\le\varepsilon_0(u)$.
Choose a particular $\varepsilon>0$ satisfying $\varepsilon \le \min_u \varepsilon_0(u)$ and $\varepsilon < \frac{\delta}{2}$. 
Note that $\delta$ and $\varepsilon$ are not computed explicitly; 
we only fix such values conceptually.


Now we define the capped subpolygon $Q_u^\varepsilon := Q_u \cup T_u(\varepsilon)$, which is a simple polygon.
Figure~\ref{fig:bridge.node.cap} illustrates how the cap $T_u(\varepsilon)$ attaches to $e_u$ for each node type.
Let $\nopt^\varepsilon(Q_u)$ denote the minimum number of strips 
in an optimal strip partition of $Q_u^\varepsilon$.
Among all optimal strip partitions of $Q_u^\varepsilon$, 
let $\mathcal{S}^\varepsilon(Q_u,e_u)$ be the set of strips incident to the apex $q_u(\varepsilon)$, and define
$\iant^\varepsilon(Q_u,e_u) := \min_{\subseteq}\{\projx{S} \mid S \in \mathcal{S}^\varepsilon(Q_u,e_u)\}$, 
where $\projx{S}$ is the interval obtained by projecting $S$ onto the $x$-axis.
Then $\iant^\varepsilon(Q_u,e_u)$ is an antichain of intervals.
By definition, every interval in $\iant^\varepsilon(Q_u,e_u)$ is realized as the $x$-projection of a strip incident to $q_u(\varepsilon)$ in some optimal partition of $Q_u^\varepsilon$.
Although the family $\mathcal{S}^\varepsilon(Q_u,e_u)$ may be infinite, 
Algorithm~\ref{alg:update-node} and Lemma~\ref{lem:num.meet.join} together imply that $\iant^\varepsilon(Q_u,e_u)$ is finite.
We define the dynamic-programming subproblem 
$\nnprob^\varepsilon(Q_u,e_u)$ to be the task of computing both $\nopt^\varepsilon(Q_u)$ and $\iant^\varepsilon(Q_u,e_u)$.
The pair
$\bigl(\nopt^\varepsilon(Q_u),\,\iant^\varepsilon(Q_u,e_u)\bigr)$
is called the \emph{DP state} at node $u$; its first entry is the
\emph{value component}, and its second entry is the
\emph{antichain component}.

\section{Dynamic programming algorithm and complexity}\label{sec:dynamic.programming.alg}
We process the nodes of the rooted binary tree $\widetilde G$ in bottom-up order.
For each node $u$, after the subproblems for all of its descendants have been
solved, we compute the DP state of $\nnprob^\varepsilon(Q_u,e_u)$ using the update
rule for either a trapezoid node or a bridge node.

During this process, we maintain the following invariant for each node $u$ of $\widetilde G$.
The value $\nopt^\varepsilon(Q_u)$ is the minimum number of strips in a strip partition of $Q_u^\varepsilon$.
Let $\mathcal{S}^\varepsilon(Q_u,e_u)$ be the set of strips incident to the apex $q_u(\varepsilon)$ 
among all optimal partitions of $Q_u^\varepsilon$.
Then $\iant^\varepsilon(Q_u,e_u)$ consists of the inclusion-minimal $x$-intervals of the strips in $\mathcal{S}^\varepsilon(Q_u,e_u)$.

\subsection{Bottom-up evaluation on the refined tree}
The update procedure is described in detail in Algorithm~\ref{alg:update-node}.
Throughout the algorithm, we assume that $\normal(e_u) = \textsf{right}$; the case $\normal(e_u) = \textsf{left}$ is symmetric.
We use the low-pass operator $\filter$ on the family of intervals.
For $\mathcal{F} \subseteq \cinterval_\mathbb{R}$, 
let $\filter(\mathcal{F}):= \{\, I \in \mathcal{F} \mid |I| \le 1 \,\}$, 
which discards intervals of length greater than~$1$.

\subparagraph{Trapezoid nodes.}
Let $u$ be a trapezoid node corresponding to a trapezoid $\Delta$.
If $u$ has a child $v$ in $\widetilde G$, 
the DP state $(\nopt^\varepsilon(Q_v), \iant^\varepsilon(Q_v,e_v))$ 
has already been computed.
If $u$ is a leaf, 
we introduce a virtual child $v$ by taking
$e_v$ to be the vertical side of $\Delta$ opposite to $e_u$, letting $Q_v$ be
the degenerate trapezoid consisting of $e_v$, and setting
$\normal(e_v)=\normal(e_u)$; then 
$\nopt^\varepsilon(Q_v)=1$ and 
$\iant^\varepsilon(Q_v,e_v)$ consists of a single interval of length $\varepsilon$ adjacent to $x(e_v)$.

\begin{figure}[t!]
  \centering
  \includegraphics[width=\textwidth]{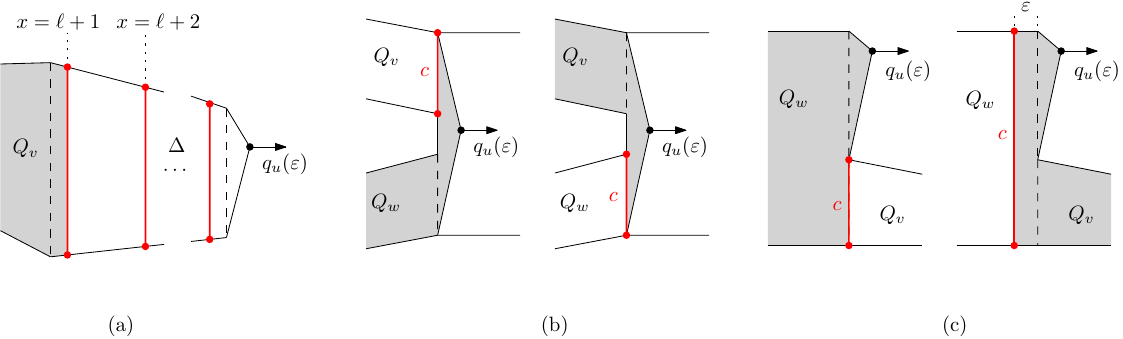}
  \caption{\small Additional vertical cuts are inserted 
  when no strip can be merged or extended across the node:
(a) trapezoid node, (b) same-side bridge node, and (c) cross-side bridge node.}
  \label{fig:placing.cuts}
\end{figure}

The update at a trapezoid node proceeds as follows.
Let $[a,b]$ be the $x$-interval of $\Delta$, where $a=x(e_v)$ and $b=x(e_u)$.
For each strip $S \in \mathcal{S}^\varepsilon(Q_v,e_v)$, 
we remove the old cap $T_v(\varepsilon)$, insert $\Delta$ into the polygon, 
and then attach the new cap $T_u(\varepsilon)$ along $e_u$. 
If $[\ell,r]$ is the $x$-interval of $S$, then the extended strip has $x$-interval
$[\min\{\ell,a\},\,\max\{r,b+\varepsilon\}]$, which is exactly the interval obtained 
by taking the meet of $\{[\ell,r]\}$ with $\{[a,b+\varepsilon]\}$.
Thus, the inclusion-minimal $x$-intervals of all feasible extensions are exactly
$\filter\bigl(\iant^\varepsilon(Q_v,e_v)\wedge \{[a,b+\varepsilon]\}\bigr)$.


If at least one extended strip survives, then we do not introduce any new cut at $u$.
If none survive, 
then no strip can be extended through $\Delta$ without violating the unit-width constraint.
Let $[\ell,r]$ be the $x$-interval of the rightmost strip in
$\mathcal{S}^\varepsilon(Q_v,e_v)$. 
If no extended strip survives, we insert a series of vertical cuts in $\Delta \cup T_u(\varepsilon)$ at $x$-coordinates $\ell+1, \ell+2, \ldots$, decomposing the extended region into
unit-width strips (with the last piece possibly shorter than~$1$).
We then increase $\nopt^\varepsilon(Q_u)$ by the number of new strips, and define
$\iant^\varepsilon(Q_u,e_u)$ to consist of the single $x$-interval 
of the rightmost suffix strip incident to $q_u(\varepsilon)$.

\subparagraph{Bridge nodes.}
Let $u$ be a bridge node of $\widetilde G$ with children $v$ and $w$.
At a bridge node, we consider pairs of strips
$S_v \in \mathcal{S}^\varepsilon(Q_v,e_v)$ and
$S_w \in \mathcal{S}^\varepsilon(Q_w,e_w)$,
and ask whether they can be combined into a single strip incident to
$q_u(\varepsilon)$.
By the choice of $\varepsilon$, 
the $x$-interval of each of $S_v$ and $S_w$ already contains
$[x(e_u)-\varepsilon,\,x(e_u)+\varepsilon]$.
Hence, replacing the child caps $T_v(\varepsilon),T_w(\varepsilon)$ with the cap
$T_u(\varepsilon)$ does not change either $x$-interval.

For a same-side bridge, 
this combination is obtained by removing the caps
$T_v(\varepsilon)$ and $T_w(\varepsilon)$, gluing the two strips along $e_u$,
and attaching the cap $T_u(\varepsilon)$.
For a cross-side bridge, 
we remove the two caps, take the union of the resulting strips, and 
attach $T_u(\varepsilon)$ along $e_u$. 
In either case, if $[\ell,r]$ and $[\ell',r']$ are the $x$-intervals of $S_v$ and $S_w$, respectively,
then the combined region has $x$-interval $[\min\{\ell,\ell'\},\,\max\{r,r'\}]$.
Thus, the inclusion-minimal $x$-intervals of all feasible combinations are exactly
$\filter\bigl(\iant^\varepsilon(Q_v,e_v)\wedge \iant^\varepsilon(Q_w,e_w)\bigr)$.


If at least one pair of strips yields a feasible strip,
then no additional cut is introduced at $u$, and 
the feasible combined strips form
$\mathcal{S}^\varepsilon(Q_u,e_u)$.
Otherwise, we insert a vertical cut near $e_u$ so that 
strips from the two child subproblems 
cannot be combined across $u$. The strips from $Q_v$ and $Q_w$ 
are then extended separately to $q_u(\varepsilon)$.
Hence, $\iant^\varepsilon(Q_u,e_u) = \iant^\varepsilon(Q_v,e_v)\vee \iant^\varepsilon(Q_w,e_w)$.


\begin{algorithm}[t!]
  \caption{DP update at a node $u$ (case $\normal(e_u)=\textsf{right}$)}
  \label{alg:update-node}
  \DontPrintSemicolon

  \KwIn{node $u$ of $\widetilde G$ and the DP states at its children}
  \KwOut{DP state $(\nopt^\varepsilon(Q_u),\, \iant^\varepsilon(Q_u,e_u))$ at $u$}

  \uIf{$u$ is a trapezoid node with child $v$}{
    \textsc{TrapezoidUpdate}$(u,v)$\;
  }
  \ElseIf{$u$ is a bridge node with children $v,w$}{
    \textsc{BridgeUpdate}$(u,v,w)$\;
  }

  \BlankLine
  \SetKwProg{Proc}{procedure}{}{end}

  \Proc{\textsc{TrapezoidUpdate}$(u,v)$}{
    Let $\Delta$ be the trapezoid of $u$ with $x$-interval $[a,b]$\;
    $\iant^\varepsilon(Q_u,e_u)
      \gets \filter\bigl(\iant^\varepsilon(Q_v,e_v) \wedge \{[a,b+\varepsilon]\}\bigr), \quad
      \nopt^\varepsilon(Q_u) \gets \nopt^\varepsilon(Q_v)$\;
    \If{$\iant^\varepsilon(Q_u,e_u) = \varnothing$}{
      let $[\ell,r]$ be the interval in $\iant^\varepsilon(Q_v,e_v)$
        with maximum left endpoint $\ell$\;
    $k \gets \lfloor b+\varepsilon - \ell \rfloor$\;
    $\iant^\varepsilon(Q_u,e_u) \gets \{[\ell + k,\ b+\varepsilon]\},\quad
    \nopt^\varepsilon(Q_u) \gets \nopt^\varepsilon(Q_v) + k\;$
    }
  }

  \BlankLine

  \Proc{\textsc{BridgeUpdate}$(u,v,w)$}{
    $\iant^\varepsilon(Q_u,e_u)
      \gets \filter\bigl(\iant^\varepsilon(Q_v,e_v) \wedge \iant^\varepsilon(Q_w,e_w)\bigr)$\;
    $\nopt^\varepsilon(Q_u) \gets \nopt^\varepsilon(Q_v) + \nopt^\varepsilon(Q_w)-1$\;

    \If{$\iant^\varepsilon(Q_u,e_u) = \varnothing$}{
      $\iant^\varepsilon(Q_u,e_u)
        \gets \iant^\varepsilon(Q_v,e_v) \vee \iant^\varepsilon(Q_w,e_w)$\;
      $\nopt^\varepsilon(Q_u) \gets \nopt^\varepsilon(Q_v) + \nopt^\varepsilon(Q_w)$\;
    }
  }
\end{algorithm}
\subparagraph{Cut placement at nodes.}
We now describe how new cuts are placed when strips cannot be merged or extended at a node $u$
without violating the unit-width constraint.

If $u$ is a trapezoid node, no strip extends through its trapezoid $\Delta$, 
so we place vertical cuts in $\Delta \cup T_u(\varepsilon)$.
They start at $x=\ell+1$, where $[\ell,r]$ is the interval with the maximum left
endpoint in $\iant^\varepsilon(Q_v,e_v)$, and continue at unit distance.
See Figure~\ref{fig:placing.cuts}(a).

Now suppose that $u$ is a bridge node with children $v$ and $w$.
If some strip in $\mathcal{S}^\varepsilon(Q_v,e_v)$ can merge with one in
$\mathcal{S}^\varepsilon(Q_w,e_w)$ into a unit-width strip,
no new cut is added at $u$.
Otherwise, we place a single vertical cut near $e_u$ 
to prevent any strip from $Q_v$ from merging with a strip from $Q_w$, so
that all strips are extended to $q_u(\varepsilon)$ separately.
See Figure~\ref{fig:placing.cuts}(b--c).

The position of this cut depends on the bridge type and on the child from which
each $x$-interval in $\iant^\varepsilon(Q_u,e_u)$ originates.
For a same-side bridge, the cut is placed along $e_w$ if the interval comes
from $\iant^\varepsilon(Q_v,e_v)$ and along $e_v$ if it comes from
$\iant^\varepsilon(Q_w,e_w)$.
For a cross-side bridge, either $e_v \subsetneq e_w$ or
$e_w \subsetneq e_v$. Assume $e_v \subsetneq e_w$. If the interval comes from
$\iant^\varepsilon(Q_w,e_w)$, the cut is placed along $e_v$. 
Otherwise it is placed inside
$Q_w$ at $x$-distance $\varepsilon$ from $e_u$, blocking any strip from $Q_w^\varepsilon$ from
reaching $q_u(\varepsilon)$; see Figure~\ref{fig:placing.cuts}(c). 
The case $e_w \subsetneq e_v$ is symmetric.
Note that inserting these cuts at bridge nodes 
does not increase the $x$-interval of any strip from
$Q_v^\varepsilon$ or $Q_w^\varepsilon$; 
we prove this in the proof of Lemma~\ref{lem:alg.correctness}.

\medskip

Every interval appearing in an antichain $\iant^\varepsilon(Q_u,e_u)$ has endpoints 
of the form $x(p)+m$ or $x(p)+m \pm \varepsilon$, where $p$ is a vertex of $P$ and 
$m \in \mathbb{Z}$.  
Thus all endpoints lie in the locally finite, totally ordered set 
$\mathbb{Z}(P,\varepsilon)$, and each antichain maintained by our DP is an element of 
the Clarke–Cormack–Burkowski lattice $\ccb_{\mathbb{Z}(P,\varepsilon)}$.
The update rules are written using meet ($\wedge$) and join ($\vee$) in this lattice, 
and the outputs of these operations remain in $\ccb_{\mathbb{Z}(P,\varepsilon)}$.

Since $\varepsilon>0$ is chosen sufficiently small, the values 
$\nopt(Q_u)$ and $\nopt^\varepsilon(Q_u)$ differ by at most one.
This difference occurs only when, in every optimal partition of $Q_u^\varepsilon$,
the only strip incident to $q_u(\varepsilon)$ is the cap $T_u(\varepsilon)$; 
in this case we have
$\iant^\varepsilon(Q_u,e_u) = \{[x(e_u),\, x(e_u)+\varepsilon]\}$.

In the DP state, we store an endpoint of the form $x(p)+m$ or $x(p)+m\pm\varepsilon$
as the value $x(p)+m$ together with a single bit indicating whether the 
$\varepsilon$–offset is present.
Since $\varepsilon>0$ is chosen sufficiently
small, all order and width tests used by the algorithm can be decided from the
base values and the flags.
For example, if $a,b \in \mathbb{Z}(P,0)$ and $m \in \mathbb{Z}$, then
\[
\begin{array}{rcl@{\qquad}rcl}
b-a-2\varepsilon>m &\iff& b-a>m, &
b-a+2\varepsilon>m &\iff& b-a\ge m,\\
\lfloor(b-a)+2\varepsilon\rfloor &=& \lfloor b-a \rfloor,  &
\lfloor(b-a)-2\varepsilon\rfloor &=& \lceil b-a \rceil - 1.
\end{array}
\]

\begin{lemma}\label{lem:alg.correctness}
When Algorithm~\ref{alg:update-node} is executed on the tree $\widetilde G$ in
bottom-up order, it correctly computes, for every node $u$ of $\widetilde G$, the pair
$(\nopt^\varepsilon(Q_u), \iant^\varepsilon(Q_u,e_u))$,
where $\nopt^\varepsilon(Q_u)$ is 
the number of strips in an optimal partition
of $Q_u^\varepsilon$, and $\iant^\varepsilon(Q_u,e_u)$ is the antichain of $x$-intervals
of the strips incident to $q_u(\varepsilon)$ across all optimal partitions of $Q_u^\varepsilon$.
\end{lemma}
\begin{proof}
For a virtual leaf $u$, $Q_u^\varepsilon$ is exactly the cap
$T_u(\varepsilon)$.  
Hence any optimal partition of $Q_u^\varepsilon$ is just $T_u(\varepsilon)$, 
so $\nopt^\varepsilon(Q_u)=1$ and $\iant^\varepsilon(Q_u,e_u)=\{[x(e_u), x(e_u)+\varepsilon]\}$.
Starting from these base cases, 
we prove the claim by induction on the nodes $u$ of $\widetilde G$ in bottom-up order.
Without loss of generality, assume
$\normal(e_u) = \texttt{right}$, and that the statement holds for all children
of $u$ in $\widetilde G$.

First consider the case that $u$ is a trapezoid node with child $v$, 
and let $\Delta$ be the trapezoid corresponding to $u$.
Let the projection of $\Delta$ onto the $x$-axis be $[a,b]$, where $a = x(e_v)$ and
$b = x(e_u)$.
Since $Q_u^\varepsilon = Q_v \cup \Delta \cup T_u(\varepsilon)$, every strip in
$\mathcal{S}^\varepsilon(Q_v,e_v)$ is extended through $\Delta$ up to $q_u(\varepsilon)$.
In terms of $x$-intervals, this corresponds to taking the meet of
$\iant^\varepsilon(Q_v,e_v)$ with $\{[a,b+\varepsilon]\}$ and removing intervals
of length greater than~$1$, that is,
$\mathcal{F}:=\filter\bigl(\iant^\varepsilon(Q_v,e_v) \wedge \{[a,b+\varepsilon]\}\bigr)$.

If $\mathcal{F}\neq\varnothing$,
then at least one strip can be extended through $\Delta$ up to $q_u(\varepsilon)$
without violating the width constraint, 
no additional cut is required, and the antichain component at $u$ is exactly $\mathcal{F}$.

Otherwise, no strip can be extended under the width constraint.
Let $[\ell,r]$ be the interval in $\iant^\varepsilon(Q_v,e_v)$ with maximum left endpoint
$\ell$. We insert vertical cuts in $\Delta \cup T_u(\varepsilon)$ at
$x = \ell+1, \ell+2, \ldots$, up to $b+\varepsilon$, 
obtaining $\lfloor b+\varepsilon-\ell\rfloor$ new strips;
this choice minimizes 
the horizontal width of the strip incident to $q_u(\varepsilon)$.
The strip incident to $q_u(\varepsilon)$ has $x$-interval $[\ell+\lfloor b+\varepsilon-\ell\rfloor,b+\varepsilon]$.
This exactly matches the trapezoid update in Algorithm~\ref{alg:update-node}.

Now suppose that $u$ is a same-side bridge node with children $v$ and $w$.
Then $e_v$ and $e_w$ lie on the same vertical line, $\normal(e_v) = \normal(e_w)$, and
$x(e_u) = x(e_v) = x(e_w)$.
Consider a strip incident to $q_u(\varepsilon)$ that is obtained by merging one strip from
$\mathcal{S}^\varepsilon(Q_v,e_v)$ and one from $\mathcal{S}^\varepsilon(Q_w,e_w)$.
Its $x$-interval is the smallest interval containing the two corresponding
$x$-intervals, and it is feasible only if its length is at most~$1$.
Thus the family of $x$-intervals of all such merged strips is
\[
\mathcal{F} := \filter\bigl(\iant^\varepsilon(Q_v,e_v) \wedge \iant^\varepsilon(Q_w,e_w)\bigr).
\]

If the antichain $\mathcal{F}$ is nonempty, then at least one pair of strips can be
merged into a strip incident to $q_u(\varepsilon)$, and we have
\[
\iant^\varepsilon(Q_u,e_u) = \mathcal{F},\qquad
\nopt^\varepsilon(Q_u) = \nopt^\varepsilon(Q_v) + \nopt^\varepsilon(Q_w) - 1.
\]
In particular, $\iant^\varepsilon(Q_v,e_v)$ (or $\iant^\varepsilon(Q_w,e_w)$)
already consists of the singleton interval
$[x(e_u),x(e_u)+\varepsilon]$ if and only if 
the corresponding cut along $e_v$ (or $e_w$) is
already present. 
In that case, $\mathcal{F}$ is nonempty,
no additional cut is introduced, and the above equalities apply.

If $\mathcal{F}=\varnothing$, 
then no pair of strips can be merged without violating the width constraint.
We place an additional cut either along $e_v$ or along $e_w$, so that strips from
$Q_v^\varepsilon$ and $Q_w^\varepsilon$ are extended independently to $q_u(\varepsilon)$.
Note that every interval in $\iant^\varepsilon(Q_v,e_v)$ and
$\iant^\varepsilon(Q_w,e_w)$ 
already contains $[x(e_u)-\varepsilon,x(e_u)+\varepsilon]$; 
otherwise one of the antichains is 
the singleton of the interval $[x(e_u),x(e_u)+\varepsilon]$, which would yield
$\mathcal{F}\neq\varnothing$, a contradiction.
Thus, this extension to $q_u(\varepsilon)$ 
does not change the $x$-intervals in $\iant^\varepsilon(Q_v,e_v)$ and
$\iant^\varepsilon(Q_w,e_w)$. 
The $x$-intervals of strips incident to $q_u(\varepsilon)$ are exactly the inclusion-minimal intervals in
$\iant^\varepsilon(Q_v,e_v) \cup \iant^\varepsilon(Q_w,e_w)$, that is,
\[
\iant^\varepsilon(Q_u,e_u)
 = \iant^\varepsilon(Q_v,e_v)\vee\iant^\varepsilon(Q_w,e_w),\qquad
\nopt^\varepsilon(Q_u)
 = \nopt^\varepsilon(Q_v)+\nopt^\varepsilon(Q_w).
\]
This matches the bridge update in Algorithm~\ref{alg:update-node}.



In the case that $u$ is a cross-side bridge node with children $v$ and $w$, we may
assume without loss of generality that $e_v \subsetneq e_w$.
The correctness follows by the same argument as in the same-side case.
The only difference is that, when no pair of strips can be merged under the width
constraint (i.e., $\mathcal{F}=\varnothing$), 
the additional cut blocking strips from $Q_w^\varepsilon$ is placed at
$x(e_u)-\varepsilon$ instead of along $e_w$.
This cut may force any strip in $\mathcal{S}^\varepsilon(Q_v,e_v)$ that reaches
$q_u(\varepsilon)$ to have an $x$-interval containing $[x(e_u)-\varepsilon,x(e_u)]$.
Since every interval in $\iant^\varepsilon(Q_v,e_v)$ already contains
$[x(e_u)-\varepsilon,x(e_u)+\varepsilon]$, their $x$-intervals do not change by the extension,
as in the same-side case.
\end{proof}

By Lemma~\ref{lem:alg.correctness}, Algorithm~\ref{alg:update-node} correctly
computes $(\nopt^\varepsilon(Q_\rho), \iant^\varepsilon(Q_\rho,e_\rho))$ at the root
$\rho$ of $\widetilde G$, where $Q_\rho = P$.
By the discussion above on the relation between $\nopt^\varepsilon(Q_u)$ and $\nopt(Q_u)$,
we obtain $\nopt(P)$ from this pair in constant time.

\subparagraph{Encoding cut descriptors.}
We show that the additional cuts placed at trapezoid and bridge nodes admit a lossless
encoding from which all cut descriptors in the optimal partition can be reconstructed in
$O(n)$ time using $O(n)$ extra space. (The detailed encoding scheme and reconstruction
procedure are given in the proof of Lemma~\ref{lem:cut-encoding}.) 
Importantly, we can generate this encoding during the DP without increasing the asymptotic running time.

\begin{lemma}\label{lem:cut-encoding}
The cut descriptors produced by the dynamic program admit an $O(n)$-space encoding
from which all cuts in an optimal partition can be reconstructed in $O(n)$ time.
\end{lemma}
\begin{proof}
To reconstruct all cut descriptors, we record additional information at each node of
$\widetilde G$.
Let $u$ be a node of $\widetilde G$.
The antichain $\iant^\varepsilon(Q_u,e_u)$ consists of $x$-intervals, which are the
projections of strips onto the $x$-axis.
Every endpoint of such an interval comes from either a vertical cut in $P$ or a vertical
edge on the boundary of a trapezoid in $\td$.

\subparagraph{Information stored at DP states.}
We treat each endpoint as a distinct item and associate two events with it: \emph{birth}
and \emph{death}.
An endpoint is born at the unique trapezoid node where it first appears in some
antichain $\iant^\varepsilon(Q_u,e_u)$.
This includes the leaf trapezoids that initialize the dynamic program and the trapezoid
nodes where Algorithm~\ref{alg:update-node} inserts new vertical cuts.
For each endpoint we store its birth record, that is, the trapezoid in which it was
created. There is no birth in the bridge node, because 
vertical cuts at bridge nodes do not introduce new endpoints.
Different endpoints may share the same $x$-coordinate. 
In this case, we distinguish such endpoints by their birth records, so
two endpoints with the same $x$-coordinate but different birth trapezoids are treated as distinct.
Moreover, we identify endpoints that differ only by the $\pm\varepsilon$ perturbation:
if two endpoints have the same birth trapezoid and their $x$-coordinates are $x$ and
$x \pm \varepsilon$, we regard them as the same endpoint, 
since the $\pm\varepsilon$ offset 
is only an artifact of the triangular caps used in the algorithm.

As the dynamic program proceeds in bottom-up order, the antichain at a node $u$ is
obtained from the antichains at its children by the operations in
Algorithm~\ref{alg:update-node}.
During this update some endpoints may disappear from the antichain at $u$.
Whenever an endpoint disappears at a node $u$, we store a \emph{death record} at $u$.
The death record contains the disappearing endpoint and the other endpoint of the same
interval at the moment it disappears.
If $u$ is a trapezoid node, such deaths occur exactly when
Algorithm~\ref{alg:update-node} inserts new vertical cuts inside that trapezoid.
If $u$ is a bridge node, deaths occur during the meet operation, where incident strips
are merged without introducing new cuts.

For every node $u$ we maintain a list $L_u$ of death records of endpoints that die at $u$.
For a bridge node $u$, the list $L_u$ is empty if and only if the join case of
Algorithm~\ref{alg:update-node} is performed at $u$.
For a trapezoid node $u$, the list $L_u$ is empty if and only if no new cut is added at
$u$.
Each endpoint is born at a unique trapezoid node.
At every trapezoid node, 
Algorithm~\ref{alg:update-node} creates
at most two endpoints, and there are $O(n)$ trapezoids in $\td$.
Hence,
the total number of distinct endpoints is $O(n)$, and
the total size of all lists $L_u$ over all nodes is also $O(n)$.

At bridge nodes, the cut position has two possible locations depending on which child
the surviving interval comes from.
When the join case is applied, we store one surviving interval from each child,
so that the strip represented at the node can be identified during the reconstruction phase.
In contrast, at trapezoid nodes, the positions of all cuts are fixed once the node is processed.
We store a compact description of these deterministic cuts
by recording the descriptor of the leftmost cut and the number of cuts.
These require $O(1)$ space per node, 
so the total additional space is $O(n)$.

\begin{figure}[t!]
  \centering
  \includegraphics[width=0.7\textwidth]{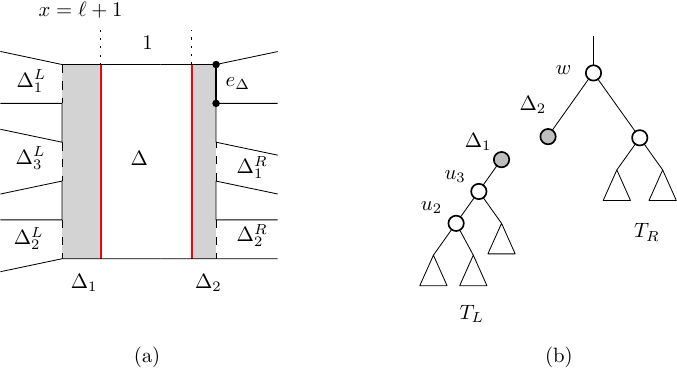}
  \caption{\small 
  Illustration of the reporting step for the example in Figure~\ref{fig:refine.binary.tree}.
  (a) Deterministic cuts inserted in $\Delta$ split it into
    subtrapezoids $\Delta_1$ and $\Delta_2$.
  (b) The cuts induce a vertex split of the node for $\Delta$ in
$\widetilde G$.
  \label{fig:reporting.cuts}}
\end{figure}

\subparagraph{Reconstruction of cut descriptors.}
From this information, we reconstruct the cut descriptors of an optimal partition of $P$.
We first reconstruct all cuts created at trapezoid nodes from their compact
descriptions and cut $P$ along these deterministic cuts, which decomposes $P$ into
subpolygons. In each resulting subpolygon, 
the remaining cuts arise only at bridge nodes where the
meet case was applied.

For each trapezoid node in $\widetilde G$ that contains deterministic cuts,
we perform a vertex split:
we replace the node by two copies, 
one for each side of the cut, 
and distribute its incident edges to these copies 
accordingly.
Then this splits $\widetilde G$ into a forest.  
See Figure~\ref{fig:reporting.cuts}.

Let $G'$ be one of the trees in this forest, and 
let $Q'$ be the subpolygon of $P$ that is covered by $G'$.
Observe that the root and all leaves of $G'$ 
are trapezoid nodes.
Each of these nodes corresponds to a subtrapezoid obtained from an original trapezoid
of the decomposition by cutting along deterministic cuts.
Let $\rho'$ be the root of $G'$, which 
is either the original root $\rho$ of $\widetilde G$ or a trapezoid node at
which new cuts were inserted.

We now choose a starting interval $I$ at the root node $\rho'$, assuming
$\normal(\rho') = \texttt{right}$.
Note that the trapezoid corresponding to $\rho'$ belongs to a
 strip in an optimal partition of $Q'$.
The starting interval $I$ is defined as the $x$-interval of this strip.

If $\rho' = \rho$, we look at the antichain $\iant^\varepsilon(Q_\rho,e_\rho)$.
If it contains an interval of length greater than~$\varepsilon$, we choose any such
interval as $I$.
Otherwise $\iant^\varepsilon(Q_\rho,e_\rho)$ consists only of the $\varepsilon$-length 
adjacent to $x(e_\rho)$, and in this degenerate case we set
$I := [x(e_\rho)-1,\,x(e_\rho)]$, which represents a strip touching $e_\rho$ from the
left.

If $\rho' \ne \rho$, then $\rho'$ is a trapezoid node where new cuts are introduced.
Among the endpoints that die at $\rho'$ due to these cuts, 
we take the rightmost left
endpoint and denote it by $\ell$.
In the update at $\rho'$, this endpoint was paired with a unique interval
$[\ell,r] \in \iant^\varepsilon(Q_v,e_v)$.
We take $I := [\ell,r]$ as the starting interval at $\rho'$.

From the starting interval $I = [\ell,r]$, 
we next find all cuts inside $Q'$ that lie on the boundary of the strip corresponding to $I$.
Let $u$ and $v$ be the birth nodes of the left and right endpoints of $I$, respectively.
Both $u$ and $v$ are leaves of $G'$.
Consider the two paths
\[
  u = u_0, u_1, \ldots, u_k = \rho'
  \quad\text{and}\quad
  v = v_0, v_1, \ldots, v_m = \rho'
\]
from $u$ and $v$ to the root $\rho'$ of $G'$.
Every internal node on these paths is either a trapezoid node or a bridge node.

If an internal node $w$ is a trapezoid node, we simply pass through
it, since no new cut inside $Q'$ is created at such a node.
Now suppose that $w$ is a bridge node that lies on one of the two paths.
We distinguish two cases according to the update at $w$.

If the join case occurred at $w$, 
then the strip represented by $I$ arrives at $w$ from
exactly one child subtree of $w$.
At that node, 
the dynamic program inserted a vertical cut that prevents this strip from merging
with strips coming from the other child. 
From the paths of the birth nodes $u$ and $v$ to $\rho'$, we know which child
$I$ comes from.
Together with the type of $w$ (same-side or cross-side), 
this determines the vertical cut at $w$. 
In the child subtree from which $I$ comes, 
we keep $I$ as the starting interval and continue the reconstruction.
In the other child subtree, we retrieve the surviving interval that was blocked by this
cut and treat it as a new starting interval.
Thus the reconstruction at $w$ produces one cut and splits the current subproblem into
two smaller subproblems, one for each child subtree.

If the meet case occurred at $w$, then $I = [\ell,r]$ was obtained at $w$ by merging two
intervals coming from the two children of $w$.
There exist intervals $I_1 = [\ell_1,r_1]$ and $I_2 = [\ell_2,r_2]$ such that
$I_1 \cup I_2 = I$.
By scanning the death records stored in $L_w$, we can recover both
$I_1$ and $I_2$, since one endpoint of each interval dies at $w$.
We then treat $I_1$ and $I_2$ as starting intervals in the two child subtrees of $w$ in $G'$ 
and apply the same reconstruction procedure recursively in each subtree.

Once the recursion reaches all leaves in all trees $G'$, every subproblem has been
processed and we have recovered every cut in an optimal strip partition of $P$.
Each node of $\widetilde G$ is processed only a constant number of times, and the total
size of all death lists $L_u$ is $O(n)$.
Therefore the reconstruction phase runs in $O(n)$ time and uses $O(n)$ additional space.
\end{proof}

\subparagraph{Failure of greedy approach.}
One might suspect that we can discard some intervals from an antichain component before passing information to its parent.
This is impossible in general:
if we discard even one interval before passing the DP state upward, 
the reduced state no longer determines the optimum.

Let $Q$ be a simple polygon and let $e$ be a vertical edge of $Q$.
Let $\normal(e)\in\{\textsf{left},\textsf{right}\}$ denote the outward normal direction of $e$.
Let $\mathcal{P}(Q,e)$ be the set of all simple polygons $P$ for which there exists a simple polygon $R$ such that
$P= Q \cup R$ and $Q \cap R\subseteq e$.
In other words, $\mathcal P(Q,e)$ consists of all simple polygons that contain $(Q,e,\normal(e))$
as a subproblem structure in the orthogonal strip partition problem.

Let $\iant(Q,e)$ be the antichain component computed by Algorithm~\ref{alg:update-node}
for the subproblem structure $(Q,e,\normal(e))$. 
A \emph{certificate} for $(Q,e,\normal(e))$ 
is a subset $\mathcal{W}\subseteq \iant(Q,e)$ 
with the following property:
there exists an algorithm that, for every $P\in\mathcal P(Q,e)$, 
computes $\nopt(P)$ from
given $\mathcal W$, $\nopt(Q)$, and the remaining region $P\setminus Q$.
The \emph{certificate complexity} of $(Q,e,\normal(e))$ is
$\min\{|\mathcal W| \mid \mathcal W \text{ is a certificate for }(Q,e,\normal(e))\}$.
This viewpoint is standard in lower bound arguments~\cite{arora_computational_2009}, 
and shows whether pruning
$\iant(Q,e)$ loses information necessary to determine $\nopt(P)$.

\begin{figure}[t!]
  \centering
  \includegraphics[width=0.75\textwidth]{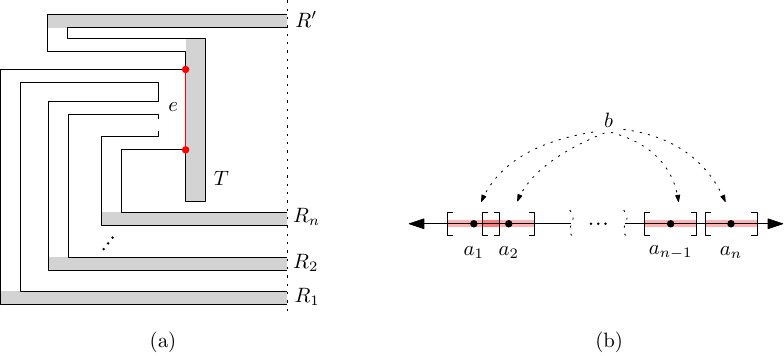}
  \caption{\small (a) 
  The construction of a simple polygon $P= Q \cup T\cup Q'$.
  The staircase rectangles $R_1,\ldots,R_n$ with their corridors form $Q$, and $R'$ with its corridor forms $Q'$; both attach to the rectangle $T$.
  (b) Determining whether $\nopt(P)=n-1$ reduces to testing whether $b$ lies within distance $\tfrac{1}{4n}$ of some $a_i$,
equivalently whether $b\in[a_i-\tfrac{1}{4n},\,a_i+\tfrac{1}{4n}]$ for some $i\in[n]$, where $|a_i-a_j|>\tfrac{1}{4n}$ for $i\neq j$.}
  \label{fig:no.greedy.construction}
\end{figure}

\begin{lemma}\label{lem:no.greedy}
For every integer $n\ge 1$, there exists a subproblem instance $(Q,e,\normal(e))$ with $O(n)$ vertices 
whose
certificate complexity is at least $n$.
\end{lemma}
\begin{proof}
Fix an integer $n \ge 1$.
For each $i = 1,\ldots, n$, choose a real number $a_i \in \left(\frac{i}{2n}-\frac{1}{4n},\frac{i}{2n}\right)$.
Let $I_i \coloneqq [a_i,a_i +1-\frac{1}{4n}]$, 
and let $\mathcal{I}$ be the family of these intervals 
$\{I_{i}\mid i = 1, \ldots, n\}$.
Note that $|I_i| = 1- \frac{1}{4n}$,
$|I_i \cup I_{i+1}| > 1$, and 
every $I_i$ contains $\frac{1}{2}$.

We realize each interval $I_i \in \mathcal{I}$ 
as a thin horizontal rectangle $R_i$ and place the rectangles at distinct heights so that they form a staircase. 
We place a thin vertical rectangle $T$ near $x=\frac{1}{2}$ above all $R_i$. 
From the upper-left corner of each $R_i$, we connect $R_i$ 
toward $T$ by a thin L-shaped corridor that goes upward and then rightward.
Just before reaching $T$, the corridors merge into a single corridor near the left side of $T$, and this corridor meets $T$ along a boundary segment.
We choose the corridors to be pairwise interior-disjoint, except at their prescribed merges near $T$.
Let $Q$ denote the union of the rectangles $R_i$ and all corridors. Then $Q$ is a simple orthogonal polygon.

Let $b \in \left(\frac{1}{2n},\frac{1}{2}\right)$ be a real number, and let $I' := [b,b+1-\frac{1}{4n}]$. 
We realize $I'$ as a thin horizontal rectangle $R'$, 
place it above $T$, and attach it to $T$
by a thin L-shaped corridor that leaves $R'$ at its lower-left corner.
Let $Q'$ denote the subpolygon consisting of $R'$ and its connecting corridor to $T$. See Figure~\ref{fig:no.greedy.construction}(a).

Let $P := Q \cup T \cup Q'$ be the simple polygon instance 
of the strip partition problem.
We root the dual tree $G$ at the trapezoid $T$. 
Then there are nodes $u$ and $u'$ 
whose subtrees induce the 
subpolygon $Q$ and $Q'$, respectively. 
By Algorithm~\ref{alg:update-node}, 
the DP state at $u$ is $(n-1,\mathcal I)$ and the state at $u'$ is $(1,\{I'\})$.
Since $|I_i\cup I'| = \bigl(1-\tfrac{1}{4n}\bigr) + |a_i-b|$, we have $\filter(\mathcal{I}\wedge \{I'\}) \neq\emptyset \iff \exists i\in[n]$ such that $|a_i-b|\le \tfrac{1}{4n}$.
Consequently, \(\nopt(P)=n-1\) if and only if $\exists i\in[n]$ with $|a_i-b|\le \tfrac{1}{4n}$; otherwise $\nopt(P)=n$. See Figure~\ref{fig:no.greedy.construction}(b).

For each $j\in [n]$, 
we can choose $(a_,\ldots, a_n, b)$ so that  
$|a_j-b| \le \tfrac{1}{4n}$ and $|a_i-b| > \tfrac{1}{4n}$ for all $i\neq j$.
Therefore, determining whether $\nopt(P)=n-1$ requires maintaining all $a_i$ (equivalently, $I_i$) values in the worst case.
In particular, the certificate complexity of $(Q,e,\normal(e))$ is $n$, where $e$ is the vertical edge of $Q$ incident to $T$.
\end{proof}

Although keeping all intervals in the antichain component often appears redundant, Lemma~\ref{lem:no.greedy} shows that 
one cannot safely prune intervals in general.
There exist instances in which every interval is necessary; 
hence 
any greedy approach that discards intervals 
using only information available in the subproblem 
fails on some instance.

\subsection{Algorithmic complexity and lower bounds}
We analyze the running time and space of our dynamic
program. 
The value component is updated only by testing whether the resulting antichain is empty, 
and the cut descriptor is computed 
during the meet or join operations without extra asymptotic cost.
Thus, the overall running time is dominated 
by the cost of the meet and join operations on antichains.

In an antichain, intervals are strictly ordered: for two distinct intervals $[a,b]$ and $[a',b']$, 
we have $a < a'$ if and only if $b<b'$.
Thus, we store each antichain as a list of intervals sorted by their left (or right) endpoints.
Every antichain produced by our lattice operations is maintained in this sorted order.

\begin{lemma}\label{lem:num.meet.join}
Let $\mathcal{A}$ and $\mathcal{B}$ be finite antichains of intervals in $\ccb_{\mathbb{Z}(P,\varepsilon)}$.
Then, $|\mathcal{A} \wedge \mathcal{B}| \le |\mathcal{A}|+|\mathcal{B}| - 1$ and
$|\mathcal{A} \vee \mathcal{B}| \le |\mathcal{A}|+|\mathcal{B}|$.
\end{lemma}

\begin{proof}
By definition of the join,
$ \mathcal{A} \vee \mathcal{B} = \min_{\subseteq}(\mathcal{A} \cup \mathcal{B})$, 
so $  |\mathcal{A} \vee \mathcal{B}|
  \le |\mathcal{A} \cup \mathcal{B}|
  \le |\mathcal{A}| + |\mathcal{B}|$.

We now bound $|\mathcal{A} \wedge \mathcal{B}|$.  
Let $\mathcal{A} = \{[\ell_i, r_i]\}_{i=1}^k$ and 
$\mathcal{B} = \{[\ell'_j, r'_j]\}_{j=1}^m$, sorted by increasing left endpoints, and
assume without loss of generality that $\ell_k \ge \ell'_m$.
Every interval in $\mathcal{A} \wedge \mathcal{B}$ arises from some pair $(i,j)$ as 
$[\min\{\ell_i,\ell'_j\}, \max\{r_i,r'_j\}]$.
If $\ell_k$ appears as a left endpoint in $\mathcal{A} \wedge \mathcal{B}$, then
$\min\{\ell_k, \ell'_j\} = \ell_k$ for some $j \in [m]$, 
which implies $\ell_k = \ell_j' = \ell_m'$ by maximality of $\ell_k$.
Since $\mathcal{A} \wedge \mathcal{B}$ is an antichain, its intervals have distinct left
endpoints, so $|\mathcal{A} \wedge \mathcal{B}| \le k + m - 1$.
\end{proof}

If $u$ is a leaf node of $\widetilde G$, 
$\iant^\varepsilon(Q_u,e_u)$ consists of a single interval. 
At every trapezoid node or bridge node, Algorithm~\ref{alg:update-node} performs exactly one meet or
one join of two antichains.
At every node $u$ of $\widetilde G$, Lemma~\ref{lem:num.meet.join} implies that
$|\iant^\varepsilon(Q_u,e_u)|$ is bounded by the size of the subtree rooted at $u$.
Boldi and Vigna~\cite{boldi_efficient_2016} give one-pass, optimally lazy algorithms 
that compute 
$\mathcal{A}\wedge\mathcal{B}$ and $\mathcal{A}\vee\mathcal{B}$ in time linear in $|\mathcal{A}|+|\mathcal{B}|$. 
Moreover, the operation $\filter$ can clearly be implemented in time linear in the number of
intervals in the antichain.

Let $T(t)$ be the total running time of 
Algorithm~\ref{alg:update-node} on a subtree with $t$ nodes, where $T(1) = O(1)$.
If $u$ is a trapezoid node, its unique child has a subtree of size $t-1$, 
the antichain at $u$ has size $O(t)$, and 
the update takes $O(t)$ time; thus $T(t) = T(t-1) + O(t)$.
If $u$ is a bridge node with children whose subtrees have sizes $a$ and $t-a-1$, 
then the two antichains together have size $O(t)$, 
so the update again takes $O(t)$ time and $T(t) = T(a) + T(t-a-1) + O(t)$.
In the worst case, this recurrence yields $T(n) = O(n^2)$.

To improve this bound, we work in the real-\texttt{RAM} model, 
where random access and binary search allow these operations to be implemented more efficiently than the one-pass algorithms.
Let $L(\mathcal{B})$ and $R(\mathcal{B})$
denote the sorted list of left and right endpoints of intervals in $\mathcal{B}$, respectively.
For each interval $I = [\ell,r] \in \mathcal{A}$, we compute the insertion ranks of
$\ell$ in $L(\mathcal{B})$ and of $r$ in $R(\mathcal{B})$, 
that is, their positions in these sorted lists.

Given two sorted lists $X$ and $Y$ with $|X|=k$ and $|Y|=m$, 
we use a selection technique of Kaplan et al.~\cite{kaplan_selection_2019} to 
compute all insertion ranks of $X$ in $Y$ in $O\left(k \log \frac{m}{k}\right)$ time for $m \ge 4k$.
We partition $Y$ into blocks of size $B := \lfloor m/(4k)\rfloor$; 
for each $x \in X$, we first identify the block containing $x$ and then perform a binary search within that block.
The block size $B$ itself can be computed in 
$O\!\left(\log \frac{m}{k}\right)$ comparisons by testing powers of two, 
so this does not change the overall running time.


The following lemma shows that 
the join and filtered meet of two antichains can be computed efficiently, 
assuming precomputed insertion ranks.
The detailed algorithm and runtime analysis are given in the proof of Lemma~\ref{lem:meet.join.algorithm}.

\begin{lemma}\label{lem:meet.join.algorithm}
Let $\mathcal{A}, \mathcal{B}$ be two finite antichains of intervals in $\ccb_{\mathbb{Z}(P,\varepsilon)}$ 
such that every interval in $\mathcal{A} \cup \mathcal{B}$ has length at most~$1$.
Assume that, for every interval $I \in \mathcal{A}$, the insertion ranks of its
endpoints among those of $\mathcal{B}$ are available in constant time.
Then, $\mathcal{A} \vee \mathcal{B}$ and $\filter(\mathcal{A} \wedge \mathcal{B})$ 
can be computed in $O(|\mathcal{A}|)$ time.
\end{lemma}
\begin{proof}
Let $\mathcal{A} = \{[\ell_i, r_i]\}_{i=1}^k$ and 
$\mathcal{B} = \{[\ell'_j, r'_j]\}_{j=1}^m$, sorted by increasing left endpoints.
Let $L(\mathcal{B}) = (\ell'_1,\dots,\ell'_m)$ and $R(\mathcal{B}) = (r'_1,\dots,r'_m)$.
For each $i \in [k]$, 
let $\alpha_i$ be the insertion rank of $\ell_i$ in $L(\mathcal{B})$, that is, the unique index $t \in \{0,\dots,m\}$ such that
$\ell'_t \le \ell_i < \ell'_{t+1}$,
where we set $\ell'_0 := -\infty$ and $\ell'_{m+1} := +\infty$.
Similarly, let $\beta_i$ be the insertion rank of $r_i$ in $R(\mathcal{B})$.

\subparagraph{Join operation.}
Consider an interval $I \in \mathcal{A}\vee\mathcal{B}:= \min_{\subseteq}(\mathcal{A}\cup \mathcal{B})$.
If $I = [\ell_i,r_i]$ comes from $\mathcal{A}$, there is no interval in $\mathcal{B}$ that is contained in $I$.
In terms of the insertion ranks $(\alpha_i,\beta_i)$ of $\ell_i$ and $r_i$, 
it is equivalent to $\alpha_i \ge \beta_i$. 
Thus, we identify all intervals in $\mathcal{A}$ that are contained in 
$ \mathcal{A}\vee\mathcal{B}$ in $O(k)$ time.

Conversely,
an interval $J = [\ell'_j, r'_j] \in \mathcal{B}$ is deleted from the join
if and only if it strictly contains some interval $[\ell_i,r_i] \in \mathcal{A}$.
Using the insertion ranks, this means $j \le \alpha_i$ and $j > \beta_i$, 
so the indices $j \in (\beta_i,\alpha_i]$ correspond to intervals of $\mathcal{B}$ that contain $[\ell_i,r_i]$ and therefore cannot belong to $\mathcal{A}\vee\mathcal{B}$.
Thus each $i$ with $\alpha_i > \beta_i$ 
defines a contiguous index range $(\beta_i,\alpha_i]$ of deleted intervals in $\mathcal{B}$.
By sweeping over the pairs $(\alpha_i,\beta_i)$, 
we obtain a family of disjoint index ranges $(\beta_i,\alpha_i]$ that cover those intervals of
$\mathcal{B}$ which must be deleted.
Instead of listing the uncovered indices one by one, we represent the complement of
$\bigcup_i (\beta_i,\alpha_i]$ as a sequence of maximal contiguous index ranges 
and take the corresponding subsequences of $\mathcal{B}$.
This takes $O(k)$ time, even though $|\mathcal{A} \vee \mathcal{B}|$ can be
as large as $k+m$ by Lemma~\ref{lem:num.meet.join}.

\subparagraph{Meet operation.}
We partition the intervals in $\mathcal{A} \wedge \mathcal{B}$ into two types:
\begin{enumerate}
\item intervals that coincide with some $J \in \mathcal{B}$, and
\item intervals whose left or right endpoint comes from some $I \in \mathcal{A}$.
\end{enumerate}

For the first type, let $J_j = [\ell'_j,r'_j] \in \mathcal{B}$.
By the definition of the meet, $J$ appears in $\mathcal{A} \wedge \mathcal{B}$ if and only if there exists $I_i = [\ell_i,r_i] \in \mathcal{A}$ with $I_i \subseteq J_j$.
In terms of the insertion ranks $(\alpha_i,\beta_i)$ of $\ell_i$ and $r_i$, 
this is equivalent to $\beta_i < j \le \alpha_i$. 
If $\beta_i < \alpha_i$, 
the indices $j$ with $\beta_i < j \le \alpha_i$ form an interval of consecutive indices.
For every $i$, we define
\[
  \mathcal{I}^{\mathcal{B}}_i :=
  \begin{cases}
    \{\, J_j \in \mathcal{B} \mid \beta_i < j \le \alpha_i \,\}, & \text{if } \beta_i < \alpha_i,\\[4pt]
    \emptyset, & \text{otherwise}.
  \end{cases}
\]
We process $i = 1,\dots,k$ in increasing order and, for each $i$, delete from
$\mathcal{I}^{\mathcal{B}}_i$ all intervals $J_j$ that already belong to some
$\mathcal{I}^{\mathcal{B}}_{i'}$ with $i' < i$.
Thus each $J_j$ is assigned to the smallest index $i$ with $\beta_i < j \le \alpha_i$,
and the subsequences $(\mathcal{I}^{\mathcal{B}}_i)_{i=1}^k$ are pairwise disjoint.
As in the join operation, they can be computed in overall $O(k)$ time.

We now consider the second type and focus on left endpoints; the case for right endpoints
is symmetric.
Scanning $\mathcal{A}$ in increasing order of left endpoints, we partition it into maximal
subsequences in which all intervals share the same insertion rank $\alpha$ of their left
endpoints.
Fix one such subsequence and write it as 
$I_i = [\ell_i,r_i]$ for $i = p,\ldots,q$.
Since 
the right endpoints $r_i$ are strictly increasing, the ranks $\beta_i$ are non-decreasing 
as $i$ increases from $p$ to $q$.
Thus, the index range $i = p,\ldots,q$ splits into up to three contiguous parts:
\[
  \{i : \beta_i < \alpha\},\quad
  \{i : \beta_i = \alpha\},\quad
  \{i : \beta_i > \alpha\}.
\]
For these three parts we have:
\begin{itemize}
  \item If $\beta_i < \alpha$, then $I_i \subseteq [\ell'_\alpha,r'_\alpha]$,  
  so $\ell_i$ cannot be the left endpoint of any interval in
        $\mathcal{A} \wedge \mathcal{B}$.
  \item If $\beta_i > \alpha$, then $[\ell'_{\alpha+1},r'_{\alpha+1}] \subseteq I_i$,
        so $I_i$ itself appears in $\mathcal{A} \wedge \mathcal{B}$.
  \item If $\beta_i = \alpha$, 
        then $\ell'_\alpha \le \ell_i < \ell'_{\alpha+1}$ and 
        $r'_\alpha \le r_i  <r'_{\alpha+1}$. 
        Each such $I_i$ yields a candidate interval $[\ell_i, r'_{\alpha+1}]$ in $\mathcal{A} \wedge \mathcal{B}$.
        All these candidates share the same right endpoint $r'_{\alpha+1}$, 
        and their left endpoints increase strictly as $i$ runs through this middle part.
        Therefore, only the candidate with the largest left
        endpoint survives as a minimal element and
        contributes a left endpoint to $\mathcal{A} \wedge \mathcal{B}$.
\end{itemize}
Thus, for a fixed left insertion rank $\alpha$ and its corresponding subsequence
$I_i = [\ell_i,r_i]$ with $i = p,\dots,q$, 
a single linear scan over $i = p,\dots,q$
computes exactly those intervals in $\mathcal{A} \wedge \mathcal{B}$ whose left endpoint
comes from $\mathcal{A}$ and whose left insertion rank is $\alpha$, in increasing order of
their left endpoints.
Repeating this for all left insertion ranks occurring in $\mathcal{A}$ yields a sequence
$\mathcal{I}_L^\mathcal{A}$ of intervals in $\mathcal{A} \wedge \mathcal{B}$, sorted by increasing
left endpoints, consisting precisely of those intervals whose left endpoint comes from an
interval of $\mathcal{A}$.
By symmetry, the same procedure applied to right endpoints yields a sequence
$\mathcal{I}_R^\mathcal{A}$, also sorted by increasing left endpoints, consisting of all intervals
whose right endpoint comes from an interval of $\mathcal{A}$.
Note that $\mathcal{I}_L^\mathcal{A}$ and $\mathcal{I}_R^\mathcal{A}$ may share intervals coming directly from
$\mathcal{A}$, and $|\mathcal{I}_L^\mathcal{A} \cup \mathcal{I}_R^\mathcal{A}| = O(k)$.
We merge $\mathcal{I}_L^\mathcal{A}$ and $\mathcal{I}_R^\mathcal{A}$ into a single sorted list
$\mathcal{I}^\mathcal{A}$, removing duplicate intervals.
Since both sequences are sorted and their total size is $O(k)$, this takes $O(k)$ time.

Then $\mathcal{A} \wedge \mathcal{B}$ is obtained by merging $\mathcal{I}^\mathcal{A}$ and
$(\mathcal{I}^\mathcal{B}_i)_{i=1}^k$ in increasing order of left endpoints, and all
resulting intervals are inclusion-minimal.
We claim that
no interval from $\mathcal{I}^\mathcal{A}$ can lie strictly between two
consecutive intervals $J_1,J_2 \in \mathcal{I}^\mathcal{B}_i$.
Suppose, for a contradiction, that 
$I \in \mathcal{I}^\mathcal{A}$ lies strictly between $J_1$ and $J_2$.
Then, $[\ell_i,r_i]\subseteq I$. 
Since $\mathcal{A}$ is an antichain, $I$ cannot be an interval of $\mathcal{A}$, so
$I$ must have at least one endpoint coming from $\mathcal{B}$.
Thus $I$ must be the smallest interval containing some $I' \in \mathcal{A}$ and $J' \in \mathcal{B}$.
As $\mathcal{B}$ is an antichain, 
$J'$ is distinct from $J_1$ and $J_2$,  
and then $J'$ must lie between them in the sorted order of $\mathcal{B}$.
This implies that $J'$ contains $[\ell_i,r_i]$, so $J' \in \mathcal{I}^\mathcal{B}_i$,
which contradicts the fact that $J_1$ and $J_2$ are consecutive in
$\mathcal{I}^\mathcal{B}_i$.

Thus each interval in $\mathcal{I}^\mathcal{A}$ is inserted between two consecutive
subsequences $\mathcal{I}^\mathcal{B}_i$ and $\mathcal{I}^\mathcal{B}_{i+1}$ (or
before the first or after the last subsequence).
Hence a single linear pass over $\mathcal{I}^\mathcal{A}$ and
$(\mathcal{I}^\mathcal{B}_i)_{i=1}^k$ yields the sorted list representing
$\mathcal{A} \wedge \mathcal{B}$.

\subparagraph{Filtering intervals in the meet.}
Note that $\mathcal{A} \wedge \mathcal{B}$ is the union of $\mathcal{I}^\mathcal{A}$ and
$\bigcup_{i=1}^k \mathcal{I}^\mathcal{B}_i$.
Every interval in $\bigcup_{i=1}^k \mathcal{I}^\mathcal{B}_i$ has length at most~$1$, so
$\filter$ only needs to be applied to $\mathcal{I}^\mathcal{A}$.
An interval $I \in \mathcal{I}^\mathcal{A}$ can have length greater than~$1$ only when it
is the union of one interval from $\mathcal{A}$ and one from $\mathcal{B}$. 
Such intervals
arise only from the middle parts.
In this case we simply discard $I$ if its length exceeds~$1$, so $\filter(\mathcal{A} \wedge
\mathcal{B})$ is obtained with no additional asymptotic cost.
\end{proof}

Let $|\mathcal{A}|=k$ and $|\mathcal{B}| = m$ with $k \le m$. 
We first compute 
the insertion ranks of all endpoints of intervals in $\mathcal{A}$ among $L(\mathcal{B})$ and $R(\mathcal{B})$
in $O\!\left(k \log \tfrac{m}{k}\right)$ time.
If $4k \le m$, 
we then compute $\filter(\mathcal{A} \wedge \mathcal{B})$ and $\mathcal{A} \vee \mathcal{B}$ in $O(k)$ time. 
If $ k \le m \le 4k$, we instead use the linear-time algorithm of Boldi and Vigna, 
which runs in $O(m+k) = O(k)$ time.

Let $T(t)$ be the running time of Algorithm~\ref{alg:update-node} on a subtree of size $t$. 
For a trapezoid node with a unique child, one antichain in the meet is a singleton, so the update takes $O(\log t)$ time and $T(t) = T(t-1) + O(\log t)$.
For a node whose children have subtree sizes $s$ and $t-s-1$ with $s \le t-s-1$, 
the cost at that node is $O\bigl(s \log \tfrac{t-s}{s}\bigr)$, so $T(t) = T(s) + T(t-s-1) + O\!\left(s \log \tfrac{t-s}{s}\right)$. These recurrences give $T(n) = O(n \log n)$ in the worst case.

\begin{theorem}
For a simple polygon $P$ 
with $n$ vertices, the problems
$\prob{V}{simp}(P)$ and $\prob{R}{simp}(P)$ can be solved in $O(n \log n)$ time using $O(n)$ space.
\end{theorem}

\subparagraph{Lower bounds.}
We now show that no sublinear-time algorithm exists for the value or reporting version, 
so our algorithm is optimal up to a logarithmic factor.
If the input polygon is given as an array of vertices in boundary order 
and each query reveals one input vertex at unit cost, 
then any algorithm requires $\Omega(n)$ time in the worst case.
\begin{theorem}\label{thm:simp.alg.optimality}
Let $P$ be a simple polygon with $n$ vertices stored in an array in boundary order,
and suppose that each query reveals one vertex of the input polygon and has unit cost.
Then, any algorithm for $\prob{V}{simp}(P)$ or $\prob{R}{simp}(P)$ requires
$\Omega(n)$ time in the worst case.
\end{theorem}
\begin{proof}
Given an input vector $\mathbf{x}\in (0,2)^n$, we construct $n+2$ vertices 
$p_i = (x_i, i)$ for $1\le i\le n$, together with 
$p_0 = (0,0)$ and $p_{n+1} = (0,n)$.
The polygon $P$ with vertices $(p_0,p_1,\ldots,p_n,p_{n+1})$ in this order is simple.
The horizontal width of $P$ is $h = \max_i{x_i}$. 
If $h \le 1$, then $P$ itself satisfies the width constraint and
$\nopt(P) = 1$; if $h > 1$, then at least two strips are required, so
$\nopt(P) > 1$. 

Suppose that an algorithm queries only $n-1$ of the vertices $p_1,\ldots,p_n$.
By an adversary argument~\cite{erickson_lower_1996}, 
the adversary replies $x_i=1/2$ to every query.
Let $j$ be the unique unqueried index.
Then two completions remain possible:
if $x_j=1/2$, then $\nopt(P)=1$;
if $x_j=3/2$, then $\nopt(P)>1$.
Thus, no correct algorithm can terminate before querying all $n$ vertices in the worst case.
Therefore, $\prob{V}{simp}(P)$ and $\prob{R}{simp}(P)$ both require $\Omega(n)$ time.
\end{proof}

\section{Orthogonal strip partitioning of self-overlapping polygons}\label{sec:extend.overlapping}
A closed curve $C\subset \mathbb{R}^2$ 
is self-overlapping if there exists a continuous map $F:D^2 \to \mathbb{R}^2$ such that 
$F|_{\mybd{D^2}}$ parametrizes $C$, and $F$ is locally one-to-one on $\myint{D^2}$. 
Intuitively, one can continuously deform a topological disk in $\mathbb{R}^2$, allowing it to overlap itself but not to tear or glue. 
The boundary of the resulting disk is a self-overlapping curve.

When $C$ is a closed polygonal chain, 
we refer to $C$ as a \emph{self-overlapping polygonal curve}.
For such a curve, 
Shor and van Wyk~\cite{shor_detecting_1992} 
define constrained Delaunay triangulations (CDTs) of $C$
whose union is an immersed disk bounded by $C$.
Each CDT determines a unique visibility relation between points on $C$.
A given curve $C$ may admit multiple CDTs, and each CDT induces a unique visibility relation on $C$.
Thus we define a \emph{self-overlapping polygon} $P$ 
to be a pair $(C,T)$, 
where $C$ is a self-overlapping polygonal curve and $T$ is a triangulation of the immersed disk bounded by $C$.
Once $(C,T)$ is fixed, it determines a visibility structure, so that strip partitions are well-defined.
See Figure~\ref{fig:self.overlap.partition}(a--b) for an illustration.

Alternatively, we can describe a self-overlapping polygon using a gluing model.
Let $T=\{\Delta_1,\Delta_2,\ldots,\Delta_m\}$ be a finite set 
of closed triangles,
and let $G=(V,E)$ be a tree with $V=T$.
For each edge $(i,j)\in E$, the triangles 
$\Delta_i$ and $\Delta_j$ are interior-disjoint and share a common segment.
We glue these shared segments 
in the topological sense, by identifying corresponding points along them. 
The gluing specified by $G$ 
produces a connected surface homeomorphic to a closed disk, 
and its boundary is a closed polygonal chain $C$.
If this surface is an immersed disk bounded by $C$, 
then $(C,T)$ defines a self-overlapping polygon, which we represent by $(T,G)$.

\begin{figure}[t!]
  \centering
  \includegraphics[width=\textwidth]{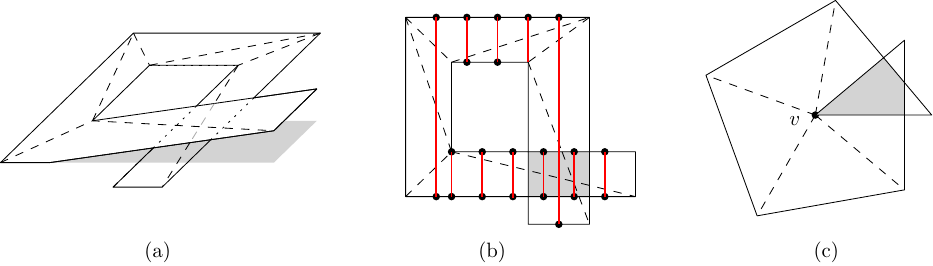}
  \caption{\small (a) A triangulated generalized terrain in $\mathbb{R}^3$ whose projection onto the
$xy$-plane defines a self-overlapping polygon. 
(b) An orthogonal strip partition of the corresponding self-overlapping polygon. 
(c) A counterexample represented by a $(T,G)$-model, but not a self-overlapping polygon.}
  \label{fig:self.overlap.partition}
\end{figure}

Given a self-overlapping polygon $P = (T,G)$, 
the visibility relation on the boundary $C$ is defined as follows.
Let $p$ and $q$ be two points on $C$ with $p\in \Delta_i, q\in\Delta_j$ 
for some triangles $\Delta_i, \Delta_j \in T$.
Let $G'$ be the minimal subtree of $G$ 
connecting $\Delta_i$ and $\Delta_j$, and let $T'$ be 
the subset of triangles corresponding to $G'$.
Gluing the triangles in $T'$ according to $G'$ yields 
a subpolygon $P'$, whose boundary contains $p$ and $q$.
We say that $p$ and $q$ are \emph{mutually visible} in $P$ 
if and only if $P'$ is a simple polygon and the segment connecting $p$ to $q$
is a cut in $P'$.

In particular, a vertical cut in $P$ is a vertical segment
whose endpoints lie on $C$ and are mutually visible in $P$.
The horizontal width of $P$  
is defined as the length of the projection of its boundary $C$ onto the $x$-axis.
With these notions, 
the problem of orthogonal strip partitioning of self-overlapping polygons 
is well-defined.

\subsection{Extending the algorithm to self-overlapping polygons.}
Let $P$ be a self-overlapping polygon with $n$ vertices, specified by $(T,G)$, 
where $T$ is a triangulation of $P$ and $G$ is the dual tree of $T$. 
From $(T,G)$, we construct a vertical trapezoidal decomposition of $P$ 
by shooting a vertical ray from each vertex 
into the interior of the immersed disk until it meets the boundary of $P$. 
Each such ray splits $P$ into two cells, 
so the resulting decomposition consists of $O(n)$ trapezoids. 

In the case of simple polygons, triangulation and trapezoidal decomposition are equivalent: each can be obtained from the other in linear time~\cite{chazelle_triangulation_1984,chazelle_triangulating_1991,fournier_triangulating_1984}.
Among several known results, 
Fournier and Montuno~\cite{fournier_triangulating_1984} give linear-time reductions in both directions between a triangulation and a horizontal trapezoidal decomposition. Their reductions 
operate purely on the triangulation (or trapezoid map) and its dual graph, without exploiting 
the simplicity of the boundary. 
Their algorithm also uses the fact 
that every triangulated simple polygon admits an ear. 
In our gluing model, the dual graph of a self-overlapping polygon is also a tree, 
so the two-ear theorem applies in exactly the same way. 
Thus, the linear-time reductions extend to self-overlapping polygons as well.

Given a triangulation of a self-overlapping polygon $P$ with $n$ vertices, 
we first convert it into the vertical trapezoidal decomposition of $P$ in linear time and then run Algorithm~\ref{alg:update-node} on this decomposition, solving $\prob{V}{self}(P)$ in $O(n \log n)$ time using $O(n)$ space.
In the simple-polygon setting, 
the simplicity assumption is not needed for correctness.
It is imposed only to ensure that all intermediate regions created by the algorithm
remain simple.
For self-overlapping polygons, we drop this restriction and only require that vertical cuts 
are consistent with the visibility relation induced by the trapezoidal decomposition, 
so the same dynamic program applies without modification.
The reporting procedure for reconstructing all cuts from the DP state also extends to
this setting without any change.


\begin{theorem}
Let $P$ be a self-overlapping polygon with $n$ vertices.
Given a triangulation of $P$, 
the problems $\prob{V}{self}(P)$ and $\prob{R}{self}(P)$ 
can be solved in $O(n \log n)$ time using $O(n)$ space.
\end{theorem}

\subparagraph{Remarks.}
The same algorithm applies, without modification, to any gluing model $(T,G)$
even when its boundary does not bound an immersed disk.
In that case, $(T,G)$ does not represent a self-overlapping polygon but rather a simply
connected affine $2$-manifold with polygonal boundary.
This more general class strictly contains self-overlapping polygons.

\subsection{Lower bounds in the algebraic computation-tree model}\label{subsec:lower.self}
Let $P$ be a self-overlapping polygon with $n$ vertices. 
Recall that $P$ is given as a pair $(T,G)$, where $T$ is a triangulation of $P$ and $G$ is its dual graph.
We show that the $\delta$-closeness 
problem can be reduced to $\prob{V}{self}(P)$ in linear time, 
and use this reduction to derive lower bounds for both $\prob{V}{self}(P)$ and $\prob{R}{self}(P)$.

\subparagraph{The $\delta$-closeness problem and its lower bound.}
The $\delta$-closeness problem is defined as follows: 
given a vector $\mathbf{x} = (x_1,\ldots,x_n)$ of real numbers and a real parameter $\delta > 0$, 
decide whether there exist indices $i \neq j$ such that $\lvert x_i - x_j \rvert \le \delta$.
In the framework of Ben-Or's lower-bound argument~\cite{ben-or_lower_1983}, 
we fix $\delta > 0$ and view this as a membership problem for a subset
$W \subset \mathbb{R}^n$, where 
\[
  W := \{\, \mathbf{x} \in \mathbb{R}^n \mid 
          \exists\, i \neq j \text{ with } |x_i - x_j| \le \delta \,\}.
\]
Thus $W$ is precisely the set of all YES-instances of the $\delta$-closeness problem.
It is known that the membership set $W$ 
has at least $n!$ disjoint connected components in $\mathbb{R}^n$~\cite{preparata_computational_1985};
they state this for the strict inequality $\lvert x_i - x_j \rvert < \delta$, but the same
argument applies to $\lvert x_i - x_j \rvert \le \delta$ as well, changing only the components from open to closed sets.
Therefore, any algebraic computation-tree deciding $\delta$-closeness must have
depth $\Omega(\log n!) = \Omega(n\log n)$.

\subparagraph{Normalization of the input $\mathbf{x}$ and $\delta$.}
For convenience in describing the reduction, we first normalize $\mathbf{x}$ and $\delta$
in linear time. By shifting and scaling, we may assume that $x_i \in [0,1]$ for all $i$
and $0 < \delta < \tfrac{1}{5}$. Consider the three intervals
\[
  I_L := [0,\delta), \qquad
  I_R := (1-\delta, 1], \qquad
  I_M := \Bigl(\tfrac{1}{2}- \tfrac{\delta}{2}, \tfrac{1}{2} + \tfrac{\delta}{2}\Bigr).
\]
If any of these intervals contains at least two numbers from $\mathbf{x}$, we can
immediately answer \texttt{YES}, since any two points in the same one of
$I_L$, $I_R$, or $I_M$ are at distance less than~$\delta$.

Otherwise, each of $I_L$, $I_R$, and $I_M$ contains at most one point of $\mathbf{x}$.
For every $x_j$ that lies in one of these three intervals, we explicitly check in
$O(n)$ time whether there exists an index $i \neq j$ with $|x_i - x_j| \le \delta$.
If such an index exists, we answer \texttt{YES}; otherwise, we safely remove $x_j$
from $\mathbf{x}$, since it cannot participate in any $\delta$-close pair.
After this normalization step, 
no $x_i$ lies in $I_L, I_R$, or $I_M$; 
either $x_i \in \bigl[\delta,\; \tfrac{1}{2}-\tfrac{\delta}{2}\bigr]$ or
$x_i \in \bigl[\tfrac{1}{2}+\tfrac{\delta}{2},\; 1-\delta\bigr]$.

Any pair taken from opposite sides has distance at least $\delta$ 
with equality only when $x_i = \tfrac{1}{2} - \tfrac{\delta}{2}$ and
$x_j = \tfrac{1}{2} + \tfrac{\delta}{2}$. 
We check in $O(n)$ time whether such a pair exists;
if so, we return \texttt{YES}. Otherwise, the two sides form independent
subinstances of the $\delta$-closeness problem. 
Without loss of generality, 
we focus on the left side.
Similarly, for each $x_i \in \{\delta,\tfrac{1}{2}-\tfrac{\delta}{2}\}$, we perform
a single $O(n)$-time scan to check whether it has a $\delta$-close neighbor:
if such a neighbor exists, we answer \texttt{YES}; otherwise, we delete $x_i$.
After $O(n)$ preprocessing, we may assume
\[
  \delta < x_i < \tfrac{1}{2}-\tfrac{\delta}{2}
  \quad\text{for all } i \in [n],
  \quad 0 < \delta < \tfrac{1}{5}.
\]

\subparagraph{Encoding $\mathbf{x}$ as a self-overlapping polygon.}
For each $i \in [n]$, we define two intervals 
\[
  A_i = \bigl[-(x_i-\tfrac{\delta}{2}),\, 1-\tfrac{\delta}{2}-x_i\bigr]
  \quad\text{and}\quad
  B_i = \bigl[-(1-\tfrac{\delta}{2}-x_i),\, x_i-\tfrac{\delta}{2}\bigr].
\]
These are disjoint intervals of equal length, $|A_i| = |B_i| = 1-\delta$, and
their union has length $\lvert A_i\cup B_i \rvert = 2-\delta-2x_i > 1$.

\begin{figure}[t!]
  \centering
  \includegraphics[width=0.85\textwidth]{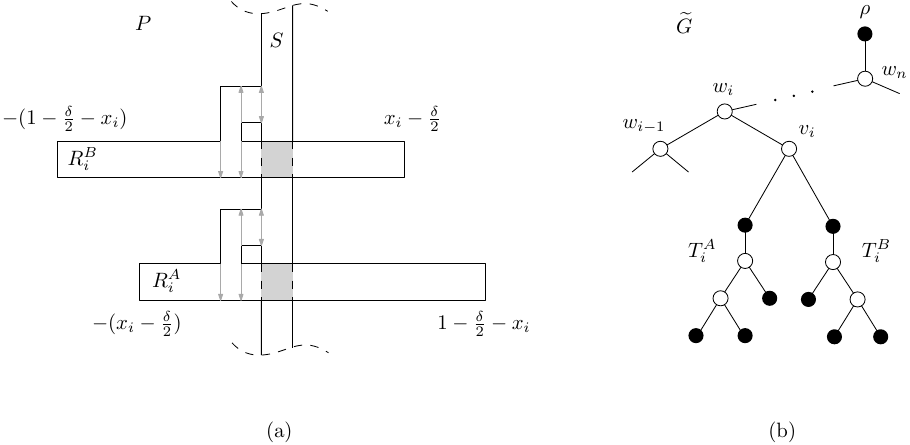}
  \caption{\small 
  Encoding the input sequence $\mathbf{x}$ as a self-overlapping polygon $P$ and its dual graph.
  (a) The rectangles $R_i^A$ and $R_i^B$ are attached to the left side of $S$ by thin $L$-shaped pipes. 
  (b) The refined binary tree $\widetilde{G}$ of the vertical trapezoidal decomposition of $P$, where the subtrees $T_i^A$ and $T_i^B$ correspond to $R_i^A$ and $R_i^B$, respectively.}
  \label{fig:self.overlap.construction}
\end{figure}

We begin by placing a long vertical rectangle $S$ in the plane whose $x$-interval is $[0,\tfrac{\delta}{2}]$.
For each index $i\in[n]$,
we create two axis-aligned rectangles $R_i^A$ and $R_i^B$ whose
$x$-projections are exactly the intervals $A_i$ and $B_i$, respectively.
We arrange these rectangles so that they are pairwise disjoint and 
appear in the vertical order
$R_1^A, R_1^B, R_2^A, R_2^B, \ldots, R_n^A, R_n^B$
from bottom to top.
Next, for each rectangle $R_i^A$ and $R_i^B$, 
we connect its interior to the left side of $S$ by thin, axis-aligned L-shaped pipes.
Gluing $S$, the rectangles $R_i^A$, $R_i^B$, and the pipes along their shared boundary segments
yields a simply connected region $P$ with $20n +4$ vertices.
See Figure~\ref{fig:self.overlap.construction}(a) for an illustration.

We represent this region using our $(T,G)$ model.
Each rectangle and pipe is triangulated independently, including vertices on shared boundary segments.
Whenever two pieces are glued along a common boundary, we take the two triangles
adjacent to that segment and add an edge between the corresponding vertices in their dual
graphs. 
In this way, 
the individual triangulations and dual graphs assemble into a single
triangulation $T$ of $P$ and its dual graph $G$.
Both $T$ and $G$ have size $O(n)$ and can be constructed in $O(n)$ time.

However, not every pair $(T,G)$ obtained by gluing arbitrary triangulated pieces
necessarily defines a self-overlapping polygon. 
For example, one can glue a
sequence of triangles so that they all meet at a single common vertex. See Figure~\ref{fig:self.overlap.partition}(c).
The resulting complex has a vertex whose neighborhood is not homeomorphic to a disk
and therefore does not bound an immersed disk, mirroring the non-self-overlapping example of Shor and van Wyk~\cite{shor_detecting_1992}.

In our construction, we can view $P$ as the projection of a generalized terrain in $\mathbb{R}^3$.
Since each rectangle $R_i^A$ and $R_i^B$ intersects $S$, 
we slightly separate $R_i^A$ and $R_i^B$ in the $z$-direction by assigning them
distinct $z$-heights and 
viewing $P$ as a polyhedral surface in $\mathbb{R}^3$. 
Following Eppstein and Mumford~\cite{eppstein_self-overlapping_2008},
a \emph{terrain} is a surface embedded in $\mathbb{R}^3$ such that every vertical line $\{(x,y,z)\mid x=c_1,y=c_2\}$ intersects it at most once. 
A generalized terrain relaxes this requirement by demanding only that every point
has a neighborhood whose image is a terrain. 
Figure~\ref{fig:self.overlap.partition}(a) illustrates a generalized terrain in $\mathbb{R}^3$.
Our layered surface satisfies this local condition, and they show that 
the projection of any generalized terrain into $xy$-plane 
has a self-overlapping curve as its boundary. 
Thus, the pair $(T,G)$ represents a self-overlapping polygon.

\subparagraph{Embedding the construction into the DP framework.}
Recall from Section~\ref{sec:simple.alg} that 
we first
compute the vertical trapezoidal decomposition of $P$ and 
refine its dual graph into a rooted binary tree $\widetilde G$.
Each node $u$ of $\widetilde G$ represents a subpolygon $Q_u$ 
with its vertical edge $e_u$ containing vertical interfaces. 
The dynamic program traverses $\widetilde G$ in bottom-up order.
For each node $u$, it computes a pair $(\nopt^\varepsilon(Q_u), \iant^\varepsilon(Q_u,e_u))$,
where $\nopt^\varepsilon(Q_u)$ is the number of strips in an optimal strip
partition of the capped subpolygon $Q_u^\varepsilon$, and
$\iant^\varepsilon(Q_u,e_u)$ is the antichain of $x$-intervals of strips
incident to $q_u(\varepsilon)$.
Algorithm~\ref{alg:update-node} specifies how this state is updated at each node,
with different rules for trapezoid nodes and bridge nodes.

Given the pair $(T,G)$ of $P$, 
we first compute the 
vertical trapezoidal decomposition of $P$ in $O(n)$ time.
For each index $i\in [n]$, 
consider the region consisting of $A_i$ (resp. $B_i$) with its incident L-shaped pipe.
The region decomposes into four trapezoids, two for $A_i$ (or $B_i$) 
and two for the pipe. 
Refining the dual graph of these four trapezoids as in Section~\ref{subsec.dynamic.approach} yields 
a binary tree whose root is a trapezoid of the pipe adjacent to $L$.
We denote these resulting trees by $T^A_i$ and $T^B_i$, respectively.
Next, for each $i$, we create 
a same-side bridge node $v_i$ whose children are the roots of $T_i^A$ and $T_i^B$.
We then connect $v_1,\ldots,v_n$ into a binary chain using same-side bridge nodes  
$w_2,\ldots w_n$, where $w_1:= v_1$ and $w_i$ has children $w_{i-1}$ and $v_i$ for $i \ge 2$.
Finally, we attach the trapezoid node $\rho$ corresponding to $L$ as a parent of $w_n$.

The resulting binary tree $\widetilde G$ is
rooted at $\rho$, and
applying the dynamic program to $\widetilde G$ correctly computes $\nopt(P)$.
The construction is illustrated in Figure~\ref{fig:self.overlap.construction}(b).

Let $a$ and $b$ be the roots of $T_i^A$ and $T_i^B$, respectively.
Since $|A_i| = |B_i| < 1$, we have
\[
  \iant^\varepsilon(Q_a,e_a) = \{A_i\}
  \quad\text{and}\quad
  \iant^\varepsilon(Q_b,e_b) = \{B_i\}.
\]
At the bridge node $v_i$, 
Algorithm~\ref{alg:update-node} computes
$\filter\bigl(\{A_i\} \wedge \{B_i\}\bigr)$,
which is empty because
$|A_i \cup B_i| > 1$.
Consequently, the update at $v_i$ uses the join, and 
the DP state at $v_i$ becomes
$\bigl(2, \{A_i,B_i\}\bigr)$.

\subparagraph{DP evaluation along the bridge chain.}
For each $i = 1,\ldots,n$, 
define $\mathcal{I}_i := \{A_i,B_i\}$.
Consider the sequence of bridge nodes $w_2,\ldots,w_n$ in $\widetilde G$ 
such that 
$w_j$ has children $v_j$ and $w_{j-1}$ with $w_1 := v_1$.
Suppose that
for $2 \le j <t$, the meet at $w_j$ contains only intervals of length $>1$, 
so each update at $w_j$ uses the join.
By induction, the antichain component at $w_{t-1}$ is precisely
the iterated join:
\[
  \iant^\varepsilon(Q_{w_{t-1}},e_{w_{t-1}})
  \;=\;
  \mathcal{I}_1 \vee \mathcal{I}_2 \vee \cdots \vee \mathcal{I}_{t-1}
  \;=\;
  \bigvee_{i=1}^{t-1} \mathcal{I}_i.
\]


In the update procedure at the bridge node $w_t$, 
the dynamic program computes $\filter((\bigvee_{i<t} \mathcal{I}_i)\wedge \mathcal{I}_t)$.
Recall that $\ccb_{\mathbb{Z}(P,\varepsilon)}$ is a completely 
distributive lattice~\cite{boldi_lattice_2018}. Hence
\[
  \Bigl(\bigvee_{i=1}^{t-1} \mathcal{I}_i\Bigr) \wedge \mathcal{I}_t
  \;=\;
  \bigvee_{i=1}^{t-1} \bigl(\mathcal{I}_i \wedge \mathcal{I}_t\bigr).
\]
Moreover, the low-pass operator $\filter(\cdot)$ is join-preserving (a
join-homomorphism), that is,
\[
  \filter(\mathcal{A} \vee \mathcal{B})
  \;=\;
  \filter(\mathcal{A}) \vee \filter(\mathcal{B})
  \quad\text{for all } \mathcal{A},\mathcal{B} \in \ccb_{\mathbb{Z}(P,\varepsilon)}.
\]
Combining these two facts, 
we can rewrite the low-pass of the meet term at $w_t$ as
\[
  \filter\!\left(
    \Bigl(\bigvee_{i=1}^{t-1} \mathcal{I}_i\Bigr) \wedge \mathcal{I}_t
  \right)
  \;=\;
  \filter\!\left(
    \bigvee_{i=1}^{t-1} \bigl(\mathcal{I}_i \wedge \mathcal{I}_t\bigr)
  \right)
  \;=\;
  \bigvee_{i=1}^{t-1} \filter\bigl(\mathcal{I}_i \wedge \mathcal{I}_t\bigr).
\]

\begin{lemma}\label{lem:delta.meet.empty}
    For $i\neq j$, 
    $\lvert x_i - x_j \rvert > \delta$ if and only if 
    $\filter(\mathcal{I}_i\wedge \mathcal{I}_j)$ is empty.
\end{lemma}
\begin{proof}
Without loss of generality assume $x_i \le x_j$ and set $d := x_j - x_i$.
By construction of $\mathcal{I}_i$ and $\mathcal{I}_j$, the meet
$\mathcal{I}_i \wedge \mathcal{I}_j$ consists of at most the following four intervals:
\[
\begin{array}{rcl@{\qquad}rcl}
\left[ -x_j+\tfrac{\delta}{2},\; 1-\tfrac{\delta}{2}-x_i \right] 
& &
&
\left[ -\!\left(1-\tfrac{\delta}{2}-x_i\right),\; x_j-\tfrac{\delta}{2} \right]
\\[6pt]
\left[ -\!\left(1-\tfrac{\delta}{2}-x_j\right),\; 1-\tfrac{\delta}{2}-x_i \right]
& &
&
\left[ -\!\left(1-\tfrac{\delta}{2}-x_i\right),\; 1-\tfrac{\delta}{2}-x_j \right]
\end{array}
\]
The lengths of these four intervals are
\[
d + 1 - \delta,\quad
d + 1 - \delta,\quad
2 - \delta - x_i - x_j,\quad
2 - \delta - x_i - x_j, \quad \text{respectively.}
\]

Since $x_i,x_j \in (\delta,\; 1/2 - \delta/2)$, we have
\[
x_i + x_j < 1 - \delta
\quad\text{and hence}\quad
2 - \delta - x_i - x_j > 1.
\tag{$\ast$}
\]

($\Rightarrow$) Suppose $\lvert x_i - x_j\rvert > \delta$, so $d > \delta$.
Then $d + 1 - \delta > 1$.
Together with $(\ast)$, this implies that every interval in
$\mathcal{I}_i \wedge \mathcal{I}_j$ has length greater than~$1$.
Since the low-pass operator $\filter$ keeps only intervals of length at most~$1$, 
$\filter(\mathcal{I}_i \wedge \mathcal{I}_j)$ is empty.

($\Leftarrow$) Conversely, suppose $\lvert x_i - x_j\rvert \le \delta$, so
$d \le \delta$.
Then $d + 1 - \delta \le 1$,
so at least the first two intervals above have length at most~$1$.
This interval survives under the low-pass operator, and we conclude that
$\filter(\mathcal{I}_i \wedge \mathcal{I}_j)$ is nonempty.
\end{proof}

By Lemma~\ref{lem:delta.meet.empty}, we have
\[
  \bigvee_{i=1}^{t-1} \filter\!\bigl(\mathcal{I}_i \wedge \mathcal{I}_t\bigr)= \varnothing
  \quad\Longleftrightarrow\quad
  |x_t - x_i| > \delta\ \text{for all } i=1,\ldots,t-1 .
\]
If this iterated join is empty,
the value component at $w_t$ is the sum of the value components of its two children, 
and the antichain component becomes $\iant^\varepsilon(Q_{w_t},e_{w_t}) = \bigvee_{i=1}^{t} \mathcal{I}_i$.
If, on the other hand, this join is nonempty, 
then 
the value
component at $w_t$ is the sum of the values of its two children minus~$1$.

\subparagraph{Reduction from $\delta$-closeness to orthogonal strip partitions.}
At every bridge node $v_i$, the value component is $2$.
Let $V_t$ denote the value component at $w_t$ for $t = 1,\ldots, n$, where $w_1 = v_1$ and $V_1 = 2$.
First suppose that the $\delta$-closeness instance is a \texttt{NO}-instance; 
that is, $|x_i - x_j| > \delta$ for all $i \neq j$.
Then, at every bridge node $w_t$ with $t\ge 2$, the value component increases by $2$.  
We obtain $V_t = 2t$ for all $t$, and in particular $V_n = 2n$.

Now suppose that the instance is a \texttt{YES}-instance.
Then there exist $i \neq j$ with $|x_i - x_j| \le \delta$.
Let $j'$ be the smallest index such that $|x_{j'} - x_i| \le \delta$ for some
$i < j'$.
For every $t < j'$, we have $|x_t - x_i| > \delta$ for all $i < t$, and hence
$V_t = 2t$ as above. 
At $w_{j'}$, 
the iterated join is nonempty, so the value component
increases only by~$1$. 
Even if the value component increases by~$2$ at all subsequent bridge nodes,
we have \[
  V_n \;\le\; (2j' - 1) + 2(n - j') = 2n - 1 \;<\; 2n.
\]
Thus $\nopt(P) = 2n$ if and only if the $\delta$-closeness instance is a
\texttt{NO}-instance.

At the root node $\rho$ corresponding to the thin rectangle $S$, we take the meet of
the antichain coming from $w_n$ with 
the singleton antichain defined by $S$. 
Since the $x$-interval of $S$ is contained in both $A_i$ and $B_i$ for every
$i \in [n]$, this meet
is just the antichain component at $w_n$, and 
the value component at $\rho$ is equal to $V_n$.
Moreover, adding or removing the triangular cap
$T_\rho(\varepsilon)$ (that is, passing between $Q_\rho$ and
$Q_\rho^\varepsilon$) does not affect the number of strips in an optimal
partition, so $\nopt(P)=V_n$.

For the self-overlapping polygon $P$ constructed from
$\mathbf{x} = (x_1,\ldots,x_n)$ with $x_i\in (\delta, \tfrac{1}{2}-\tfrac{\delta}{2})$ 
and $ 0 < \delta < \frac{1}{5}$, 
we obtain
\[
  \nopt(P) = 2n
  \quad\Longleftrightarrow\quad
  |x_i - x_j| > \delta \ \text{for all } i \neq j.
\]
Thus the $\delta$-closeness problem on $n$ numbers reduces in
linear time to solving $\prob{V}{self}(P)$ for a self-overlapping polygon $P$ 
with $O(n)$ vertices.

\begin{theorem}\label{thm:self.lowerbound}
Let $P$ be a self-overlapping polygon with $n$ vertices represented in the 
$(T,G)$-model, where 
$T$ is a triangulation of $P$ and $G$ is its dual graph.
In the algebraic computation-tree model, 
any algorithm for $\prob{V}{self}(P)$ or $\prob{R}{self}(P)$ requires
$\Omega(n\log n)$ time in the worst case.
\end{theorem}

This lower bound matches the running time of the algorithm, 
hence it is optimal.

\section{Concluding remarks}
We studied orthogonal strip partitioning for convex, simple and self-overlapping polygons. 
For convex polygons, we presented 
optimal algorithms for both value and reporting versions. In particular, for the reporting version, our reporting algorithm is input-sensitive optimal with respect to the horizontal width of $P$.

For simple and self-overlapping polygons, 
we presented a dynamic programming approach that runs in $O(n\log n)$ time and uses $O(n)$ space for both the value and reporting versions.
The DP recurrences are naturally formulated on the Clarke--Cormack--Burkowski lattice from order theory.
For self-overlapping inputs, the running time is optimal in the algebraic computation-tree model; for simple polygons, it is optimal up to a logarithmic factor in the decision-tree model.

The main open problem is to
close the logarithmic gap between
the upper and lower bounds for simple polygons. 
Lemma~\ref{lem:no.greedy} rules out any greedy pruning rule, 
so in the worst case the algorithm must compute a meet or join operation 
between antichains of size $\Theta(n)$.
However, it seems difficult to construct a simple polygon in which 
such expensive lattice operations occur repeatedly,
whereas this is easy for self-overlapping inputs.
This implies 
that simplicity may limit the number of costly lattice operations.
Even if this intuition is correct, 
a binary-search based implementation still incurs $O(n\log n)$ time.
To obtain an $O(n)$-algorithm, 
one likely needs a linear-scan procedure, 
similar to that of Boldi and Vigna~\cite{boldi_efficient_2016}, 
together with an amortized analysis of the total cost of all lattice operations.

\bibliographystyle{plain}
\bibliography{mypaper}

\end{document}